\documentclass[a4paper,UKenglish,cleveref, autoref,numberwithinsect]{lipics-v2019}
\pdfoutput=1

\usepackage{tikz}	
\usepackage{multirow}
\usetikzlibrary{decorations.pathreplacing,angles,quotes}
\usetikzlibrary{arrows}
\usepackage{float}
\newcommand{\pred}{\mathrm{pred}}
\newcommand{\comp}{\mathrm{comp}}
\usepackage[colorinlistoftodos,prependcaption,textsize=tiny]{todonotes}
\usepackage{xargs}
\newcommandx{\celine}[2][1=]{\todo[inline,linecolor=green,backgroundcolor=green!25,bordercolor=green,caption={\normalsize \textbf{Celine}},#1]{\normalsize #2}}
\newcommandx{\jesper}[2][1=]{\todo[inline,linecolor=red,backgroundcolor=red!25,bordercolor=red,caption={\normalsize \textbf{Jesper}},#1]{\normalsize #2}}

\usepackage[linesnumbered,lined,boxed,commentsnumbered,noend]{algorithm2e}

\title{On the Fine-grained Parameterized Complexity of Partial Scheduling to Minimize the Makespan} 

\titlerunning{Fine-grained Parameterized Complexity of Partial Scheduling}

\author{Jesper Nederlof}{ Utrecht University, Algorithms and Complexity group, Netherlands \and\url{https://webspace.science.uu.nl/~neder003/} }{ j.nederlof@uu.nl}{}{ERC project no. 617951. and no. 853234. and NWO project no. 024.002.003}

\author{C\'eline M. F. Swennenhuis}{Eindhoven University of Technology, Combinatorial Optimization group, Netherlands \and \url{https://research.tue.nl/nl/persons/c\%C3\%A9line-swennenhuis} }{c.m.f.swennenhuis@tue.nl}{}{NWO project no. 613.009.031b, ERC project no. 617951.}

\authorrunning{J. Nederlof and C.\,M.\,F. Swennenhuis}

\Copyright{Jesper Nederlof and C\'eline M. F. Swennenhuis}



\nolinenumbers 

\hideLIPIcs  

\begin{document}

\maketitle 
\begin{abstract} 
We study a natural variant of scheduling that we call \emph{partial scheduling}: In this variant an instance of a scheduling problem along with  an integer $k$ is given and one seeks an optimal schedule where not all, but only $k$ jobs, have to be processed.

Specifically, we aim to determine the fine-grained parameterized complexity of partial scheduling problems parameterized by $k$ for all variants of scheduling problems that minimize the makespan and involve unit/arbitrary processing times, identical/unrelated parallel machines, release/due dates, and precedence constraints.
That is, we investigate whether algorithms with runtimes of the type $f(k)n^{\mathcal{O}(1)}$ or $n^{\mathcal{O}(f(k))}$ exist for a function $f$ that is as small as possible.

Our contribution is two-fold: First, we categorize each variant to be either in $\mathsf{P}$, $\mathsf{NP}$-complete and fixed-parameter tractable by $k$, or $\mathsf{W}[1]$-hard parameterized by $k$. Second, for many interesting cases we further investigate the run time on a finer scale and obtain run times that are (almost) optimal assuming the Exponential Time Hypothesis.
As one of our main technical contributions, we give an $\mathcal{O}(8^kk(|V|+|E|))$ time algorithm to solve instances of partial scheduling problems minimizing the makespan with unit length jobs, precedence constraints and release dates, where $G=(V,E)$ is the graph with precedence constraints.

\end{abstract}
\clearpage 
\setcounter{page}{1}

\section{Introduction}
Scheduling is one of the most central application domains of combinatorial optimization.
In the last decades, huge combined effort of many researchers led to major progress on understanding the worst-case computational complexity of almost all natural variants of scheduling: By now, for most of these variants it is known whether they are $\mathsf{NP}$-complete or not. 
Scheduling problems provide the context of some of the most classic approximation algorithms. For example, in the standard textbook by Shmoys and Williamson on approximation algorithms~\cite{DBLP:books/daglib/0030297} a wide variety of techniques are illustrated by applications to scheduling problems. 
See also the standard textbook on scheduling by Pinedo~\cite{Pinedo:2008:STA:1477600} for more background.

Instead of studying approximation algorithms, another natural way to deal with $\mathsf{NP}$-completeness is \emph{Parameterized Complexity} (PC). 
%
%
While the application of general PC theory to the area of scheduling has still received considerably less attention than the approximation point of view, recently its study has seen explosive growth, as witnessed by a plethora of publications (e.g.~\cite{DBLP:journals/scheduling/BessyG19,DBLP:journals/corr/abs-1911-12350, lente, DBLP:journals/mp/MnichW15, DBLP:conf/door/BevernBBKTW16, van2015interval}).
Additionally, many recent results and open problems can be found in a survey by Mnich and van Bevern~\cite{mnich2018parameterized}, and even an entire workshop on the subject was recently held~\cite{lorentz}.

In this paper we advance this vibrant research direction with a complete mapping of how several standard scheduling parameters influence the parameterized complexity of minimizing the makespan in a natural variant of scheduling problems that we call \emph{partial scheduling}.
Next to studying the classical question of whether parameterized problems are in $\mathsf{P}$, $\mathsf{FPT}$ or $W$-hard, we also follow the well-established modern perspective of `fine-grained' PC and aim at run times of the type $f(k)n^{\mathcal{O}(1)}$ or $n^{f(k)}$ for the smallest function $f$ of parameter $k$.
\subparagraph*{Partial Scheduling.}
In many scheduling problems arising in practice, the set of jobs to be scheduled is not predetermined. We refer to this as \emph{partial scheduling}. Partial scheduling is well-motivated from practice, as it arises naturally for example in the following scenarios:
\begin{enumerate}
    \item Due to uncertainties a \emph{close-horizon approach} may be employed and only few jobs out of a big set of jobs will be scheduled in a short but fixed time-window,
    \item In freelance markets typically a large database of jobs is available and a freelancer is interested in selecting only a few of the jobs to work on,
    \item The selection of the jobs to process may resemble other choices the scheduler should make, such as to outsource non-processed jobs to various external parties.
\end{enumerate}
Partial scheduling has been previously studied in the equivalent forms of 
\emph{maximum throughput scheduling}~\cite{DBLP:conf/esa/Sgall12} (motivated by the first example setting above),
\emph{job rejection}~\cite{shabtay2013survey}, \emph{scheduling with outliers}~\cite{GuptaKKS09}, \emph{job selection}~\cite{eun2017maximizing,koulamas2013note,DBLP:journals/candie/YangG07} and its special case \emph{interval selection}~\cite{DBLP:journals/mor/ChuzhoyOR06}.

In this paper, we conduct a rigorous study of the parameterized complexity of partial scheduling,
parameterized by \emph{the number of jobs to be scheduled}.
We denote this number by $k$.
While several isolated results concerning the parameterized complexity of partial scheduling do exist, this parameterization has (somewhat surprisingly) not been rigorously studied yet.\footnote{We compare the previous works and other relevant studied parameterization in the end of this section.} We address this and study the parameterized complexity of the (arguably) most natural variants of the problem. We fix as objective to minimize the makespan while scheduling at least $k$ jobs, for a given integer $k$ and study all variants with the following characteristics:
\begin{itemize}
 \item $1$ machine, identical parallel machines or unrelated parallel machines,
 \item release/due dates, unit/arbitrary processing times, and precedence constraints.
\end{itemize}
Note that a priori this amounts to $3\times 2\times 2 \times 2 \times 2 =  48$ variants.

\subsection{Our Results}
We give a classification of the parameterized complexity of these 48 variants. Additionally, for each variant that is not in $\mathsf{P}$, we give algorithms solving them and lower bounds under ETH.  
To easily refer to a variant of the scheduling problem, we use the standard three-field notation by Graham et al.~\cite{graham1979optimization}. See Section~\ref{sec:prel} for an explanation of this notation.
To accommodate our study of partial scheduling, we extend the $\alpha|\beta|\gamma$ notation as follows:

\begin{definition}
We let $k$-sched in the $\gamma$-field indicate that we only schedule $k$ out of $n$ jobs.
\end{definition}
We study the fine-grained parameterized complexity of all problems $\alpha|\beta|\gamma$, where $\alpha \in \{1,P,R\}$, the options for $\beta$ are all combinations for $r_j,prec,d_j,p_j=1$, and $\gamma$ is fixed to $\gamma = k\text{-sched},C_{\max}$.
Our results are explicitly enumerated in Table~\ref{tab:tricho}.
\newcommand{\result}[1]{\textcolor{black}{\textbf{#1}}}
\newcommand{\pline}{\cline{2-8}}
\begin{table}[H]
\scalebox{0.75}{
\hspace{-0.5em}
\begin{tabular}{l|l|l|l|l|l|l|l}
&&\multirow{2}{*}{Problem Description} & \multirow{2}{0.9in}{Parameterized Complexity in $k$}  & \multirow{2}{0.3in}{Result Type}& \multicolumn{2}{c|}{Lower Bound under ETH } & \multirow{2}{0.45in}{Run Time}  \\
\cline{6-7}
&& &   &  & Excluded Run Time & Reduction from &  \\ \hline \hline

\multirow{20}{0.1in}{\rotatebox[origin=c]{90}{Precedence Relations}}&1&$1|  \text{prec},  p_j=1  |\gamma$   & $\mathsf{P}$  & \result{[A]} & & &$n^{\mathcal{O}(1)}$\\\pline
&2&$1|  r_j,    \text{prec},  p_j=1  |\gamma$   & $\mathsf{P}$ & \result{[A]} && &$n^{\mathcal{O}(1)}$ \\\pline
&3&$1| d_j,  \text{prec},  p_j=1  |\gamma$   & $\mathsf{W}[1]$-hard  &\result{[B]} &$n^{o(k / \log k)}$ & 3\textsc{-Coloring}  &$n^{\mathcal{O}(k)}$\\\pline
&4&$1|  r_j,  d_j,  \text{prec},  p_j=1  |\gamma$   & $\mathsf{W}[1]$-hard  &  \result{[B]} &  $n^{o(k / \log k)}$ &3\textsc{-Coloring}&$n^{\mathcal{O}(k)}$\\\pline

&5&$P|  \text{prec},  p_j=1  |\gamma$   & $\mathsf{FPT}$ &  \result{[C]} & $\mathcal{O}^*(2^{o(\sqrt{k \log k})})$ &$P|  \text{prec},  p_j=1  |C_{\max}$& $\mathcal{O}^*(2^{\mathcal{O}(k)})$ \\\pline
&6&$P|  r_j, \text{prec},  p_j=1  |\gamma$   & $\mathsf{FPT}$ &  \result{[C]} &  $\mathcal{O}^*(2^{o(\sqrt{k \log k})})$ &$P|  \text{prec},  p_j=1  |C_{\max}$& $\mathcal{O}^*(2^{\mathcal{O}(k)})$\\\pline
&7&$P|  d_j,  \text{prec},  p_j=1  |\gamma$   & $\mathsf{W}[1]$-hard  &  \result{[B]} &  $n^{o(k / \log k)}$ &3\textsc{-Coloring}&$n^{\mathcal{O}(k)}$\\\pline
&8&$P|  r_j,  d_j, \text{prec},  p_j=1  |\gamma$ &   $\mathsf{W}[1]$-hard  &  \result{[B]} & $n^{o(k / \log k)}$ &3\textsc{-Coloring}&$n^{\mathcal{O}(k)}$\\ \pline

&9&$1| \text{prec} |\gamma$   & $\mathsf{W}[1]$-hard& \result{[D]} & $n^{o(\sqrt{k})}$ &$k$\textsc{-Clique}&$n^{\mathcal{O}(k)}$ \\\pline
&10&$1|  r_j, \text{prec}  |\gamma$  & $\mathsf{W}[1]$-hard  &  \result{[D]} & $n^{o(k / \log k)}$ &\textsc{Partitioned S.I.}&$n^{\mathcal{O}(k)}$ \\\pline
&11&$1| d_j,  \text{prec}  |\gamma$   & $\mathsf{W}[1]$-hard  &  \result{[D]} &  $n^{o(k / \log k)}$ &\textsc{Partitioned S.I.}&$n^{\mathcal{O}(k)}$ \\\pline
&12&$1|  r_j,  d_j,  \text{prec} |\gamma$   & $\mathsf{W}[1]$-hard  &  \result{[D]} &  $n^{o(k / \log k)}$ &\textsc{Partitioned S.I.}&$n^{\mathcal{O}(k)}$ \\\pline

&13&$P| \text{prec} |\gamma$ &   $\mathsf{W}[1]$-hard&  \result{[D]} & $n^{o(k / \log k)}$ &\textsc{Partitioned S.I.}&$n^{\mathcal{O}(k)}$ \\\pline
&14&$P|  r_j, \text{prec}  |\gamma$   & $\mathsf{W}[1]$-hard  &  \result{[D]} &  $n^{o(k / \log k)}$ &\textsc{Partitioned S.I.}&$n^{\mathcal{O}(k)}$\\\pline
&15&$P| d_j,  \text{prec} |\gamma$   & $\mathsf{W}[1]$-hard & \result{[D]} & $n^{o(k / \log k)}$  &\textsc{Partitioned S.I.}&$n^{\mathcal{O}(k)}$ \\\pline
&16&$P| r_j, d_j,  \text{prec} |\gamma$   & $\mathsf{W}[1]$-hard  &  \result{[D]} & $n^{o(k / \log k)}$  &\textsc{Partitioned S.I.}&$n^{\mathcal{O}(k)}$ \\\pline

&17&$R| \text{prec} |\gamma$ &   $\mathsf{W}[1]$-hard &  \result{[D]} & $n^{o(k / \log k)}$ &\textsc{Partitioned S.I.}&$n^{\mathcal{O}(k)}$\\\pline
&18&$R|  r_j, \text{prec}  |\gamma$   & $\mathsf{W}[1]$-hard &  \result{[D]} &  $n^{o(k / \log k)}$ &\textsc{Partitioned S.I.}&$n^{\mathcal{O}(k)}$\\\pline
&19&$R| d_j,  \text{prec} |\gamma$   & $\mathsf{W}[1]$-hard &  \result{[D]} &  $n^{o(k / \log k)}$ &\textsc{Partitioned S.I.}&$n^{\mathcal{O}(k)}$ \\\pline
&20&$R| r_j, d_j,  \text{prec} |\gamma$   & $\mathsf{W}[1]$-hard &  \result{[D]} &$n^{o(k / \log k)}$ &\textsc{Partitioned S.I.}&$n^{\mathcal{O}(k)}$ \\
\hline
\hline	
	

\multirow{20}{0.1in}{\rotatebox[origin=c]{90}{No Precedence Relations}}&21&$1| p_j=1  |\gamma$   & $\mathsf{P}$ &  \result{[E]} && &$n^{\mathcal{O}(1)}$ \\\pline
&22&$1|  r_j,  p_j=1  |\gamma$   & $\mathsf{P}$ &  \result{[E]} &&  &$n^{\mathcal{O}(1)}$\\\pline
&23&$1|  d_j,  p_j=1  |\gamma$   & $\mathsf{P}$ &  \result{[E]} &&  &$n^{\mathcal{O}(1)}$\\\pline
&24&$1|  r_j,  d_j,  p_j=1  |\gamma$   & $\mathsf{P}$ &  \result{[E]} && &$n^{\mathcal{O}(1)}$ \\\pline

&25&$P| p_j=1  |\gamma$   & $\mathsf{P}$ &  \result{[E]} && &$n^{\mathcal{O}(1)}$ \\\pline
&26&$P| r_j, p_j=1  |\gamma$   & $\mathsf{P}$ &  \result{[E]} &&  &$n^{\mathcal{O}(1)}$\\\pline
&27&$P| d_j,  p_j=1  |\gamma$   & $\mathsf{P}$ &  \result{[E]} &&  &$n^{\mathcal{O}(1)}$\\\pline
&28&$P|  r_j,  d_j,  p_j=1  |\gamma$   & $\mathsf{P}$ &  \result{[E]} &&  &$n^{\mathcal{O}(1)}$\\\pline

&29&$1||\gamma$ &   $\mathsf{P}$ & \result{[F]} &&  &$n^{\mathcal{O}(1)}$\\\pline
&30&$1| r_j  |\gamma$ &  $\mathsf{P}$ &  \result{[F]} && &$n^{\mathcal{O}(1)}$ \\\pline
&31&$1|  d_j |\gamma$ &  $\mathsf{P}$ & \result{[F]} &&  &$n^{\mathcal{O}(1)}$  \\\pline
&32&$1|  r_j,  d_j  |\gamma$ &   $\mathsf{FPT}$ &  \result{[G]} & $\mathcal{O}^*(2^{o(k)})$ &\textsc{Subset Sum}&$\mathcal{O}^*(2^{\mathcal{O}(k)})$ \\\pline

&33&$P|  |\gamma$ &   $\mathsf{FPT}$ &  \result{[G]} & $\mathcal{O}^*(2^{o(k)})$ &\textsc{Subset Sum} &$\mathcal{O}^*(2^{\mathcal{O}(k)})$ \\\pline
&34&$P|  r_j |\gamma$   & $\mathsf{FPT}$ &  \result{[G]} & $\mathcal{O}^*(2^{o(k)})$ &\textsc{Subset Sum}&$\mathcal{O}^*(2^{\mathcal{O}(k)})$ \\\pline
&35&$P|  d_j  |\gamma$   & $\mathsf{FPT}$ &  \result{[G]} & $\mathcal{O}^*(2^{o(k)})$ &\textsc{Subset Sum}&$\mathcal{O}^*(2^{\mathcal{O}(k)})$ \\\pline
&36&$P|  r_j, d_j  |\gamma$   & $\mathsf{FPT}$ & \result{[G]} & $\mathcal{O}^*(2^{o(k)})$ &\textsc{Subset Sum}&$\mathcal{O}^*(2^{\mathcal{O}(k)})$\\\pline

&37&$R|  |\gamma$ &   $\mathsf{FPT}$  &  \result{[G]} & $\mathcal{O}^*(2^{o(k)})$ &\textsc{Subset Sum}&$\mathcal{O}^*(2^{\mathcal{O}(k)})$ \\\pline
&38&$R|  r_j |\gamma$   & $\mathsf{FPT}$ &  \result{[G]} & $\mathcal{O}^*(2^{o(k)})$ &\textsc{Subset Sum}&$\mathcal{O}^*(2^{\mathcal{O}(k)})$ \\\pline
&39&$R|  d_j  |\gamma$   & $\mathsf{FPT}$ &  \result{[G]} & $\mathcal{O}^*(2^{o(k)})$ &\textsc{Subset Sum}&$\mathcal{O}^*(2^{\mathcal{O}(k)})$ \\\pline
&40&$R|  r_j, d_j  |\gamma$   & $\mathsf{FPT}$ &  \result{[G]} & $\mathcal{O}^*(2^{o(k)})$ &\textsc{Subset Sum}&$\mathcal{O}^*(2^{\mathcal{O}(k)})$\\\pline
\end{tabular}
}
\caption{
	The fine-grained parameterized complexity of partial scheduling, where $\gamma$ denotes $k$-sched, $C_{\max}$ and \textsc{S.I.} abbreviates  \textsc{Subgraph Isomorphism}. Since $p_j=1$ implies that the machines are identical, the mentioned number of $48$ combinations reduces to $40$ different scheduling problems. The $\mathcal{O}^*$ notation omits factors polynomial in the input size.
}
\label{tab:tricho}
\end{table}

The rows of Table~\ref{tab:tricho} are lexicographically sorted on (i) precedence relations / no precedence relations, (ii) a single machine, identical machines or unrelated machines (iii) release dates and/or deadlines.
Because their presence has a major influence on the character of the problem we stress the distinction between variants with and without \emph{precedence constraints}.\footnote{A precedence constraint $a\prec b$ enforces that job $a$ needs to be finished before job $b$ can start.}
On a high abstraction level, our contribution is two-fold:
\begin{enumerate}
	\item We present a \emph{classification} of the complexity of all aforementioned variants of partial scheduling with the objective of minimizing the makespan. Specifically, we classify all variants to be either solvable in polynomial time, to be fixed-parameter tractable in $k$ and NP-hard, or to be $\mathsf{W}[1]$-hard. 
	\item For most of the studied variants we present both an algorithm and a lower bound that shows that our algorithm cannot be significantly improved unless the Exponential Time Hypothesis (ETH) fails.
\end{enumerate}
Thus, while we completely answer a classical type of question in the field of Parameterized Complexity, we pursue in our second contribution a more modern and fine-grained understanding of the best possible run time with respect to the parameter $k$.
For several of the studied variants, the lower bounds and algorithms listed in Table~\ref{tab:tricho} follow relatively quickly. However, for many other cases we need substantial new insights to obtain (almost) matching upper and lower bounds on the runtime of the algorithms solving them.
We have grouped the rows in \emph{result types} \result{[A]}-\result{[G]} depending on our methods for determining their complexity. 

\subsection{Our new Methods}

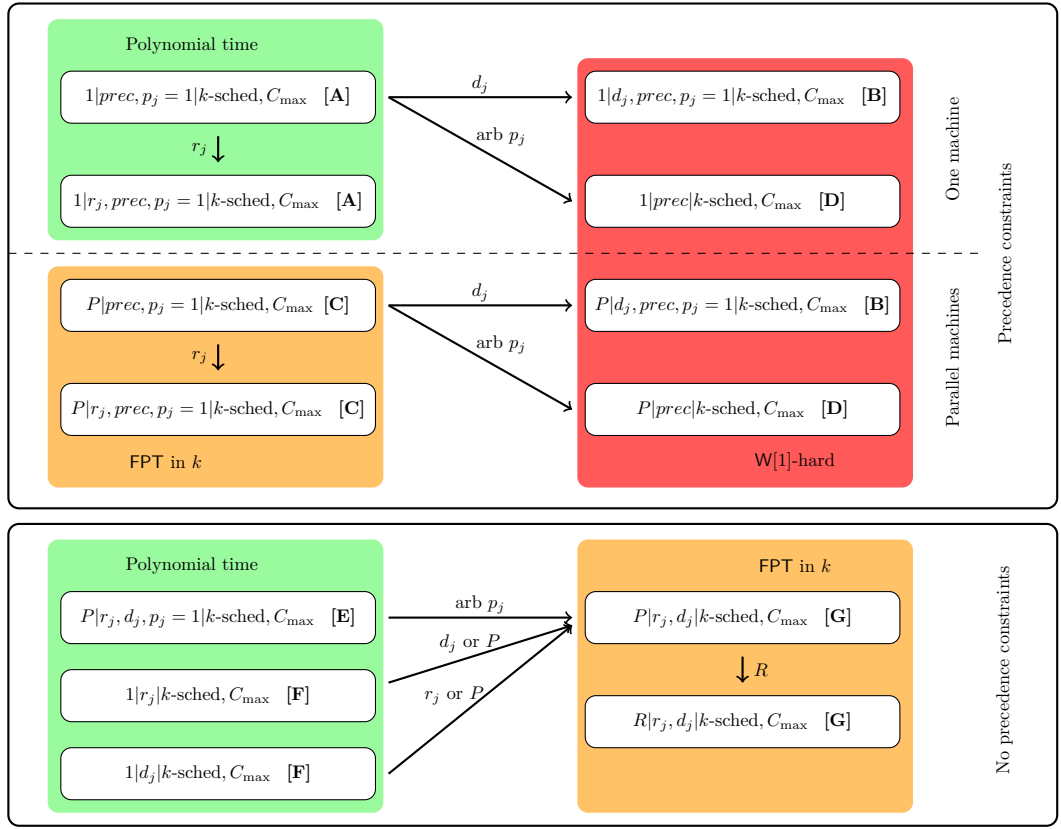
\begin{figure}[t]
    \centering
     \begin{tikzpicture} [scale = 0.69, transform shape]

        \fill [rounded corners, fill= green!95!blue!40!white] (-8.25,-4) rectangle (-1.85,-9.25);
        \node at (-5.5,-4.5) {Polynomial time};
        
        \fill [rounded corners , fill= red!40!yellow!60!white] (1.85,-4) rectangle (8.25,-9.25);
        \node at (6,-4.5) {$\mathsf{FPT}$ in $k$};
        
        \draw [rounded corners, fill = white!100!black] (-8,-7.5) rectangle (-2,-6.5) node[pos=.5] {$1 |r_j |k\text{-sched},C_{\max}$ \, \result{[F]}};
        
        \draw [rounded corners, fill = white!100!black] (-8,-9) rectangle (-2,-8) node[pos=.5] {$1 |d_j |k\text{-sched},C_{\max}$ \, \result{[F]}};
        
        \draw [rounded corners, fill = white!100!black] (-8,-6 ) rectangle (-2,-5) node[pos=.5] {$P | r_j,d_j, p_j=1|k\text{-sched},C_{\max}$ \, \result{[E]}};
        
        \draw [rounded corners, fill = white!100!black] (2,-6) rectangle (8,-5) node[pos=.5] {$P |r_j,d_j|k\text{-sched},C_{\max}$ \, \result{[G]}};
        
        \draw [rounded corners, fill = white!100!black] (2,-8) rectangle (8,-7) node[pos=.5] {$R |r_j,d_j|k\text{-sched},C_{\max}$ \, \result{[G]}};
;

        \draw[->,thick] (5,-6.25) to[] (5,-6.75);
        \node at (5.35,-6.5) {$R$};
        
        \draw[->,thick] (-1.75,-5.5) to[] (1.75,-5.5);
        \node at (0,-5.25) {arb $p_j$};
        
        \draw[->,thick]  (-1.75,-6.75) to (1.75,-5.65);
       \node at (-0.2,-6) {$d_j$ or $P$};
       
       \draw[->,thick]  (-1.75,-8.5) to (1.75,-5.65);
       \node at (-0.5,-7) {$r_j$ or $P$};

        
        
        \fill [rounded corners, fill= green!95!blue!40!white] (-8.25,6) rectangle (-1.85,1.75);
        \node at (-5.5,5.5) {Polynomial time};
        
        \fill [rounded corners , fill= red!65!white] (1.85,-3) rectangle (8.25,5.25);
        \node at (6,-2.5) {$\mathsf{W}[1]$-hard};
        
        \fill [rounded corners, fill= red!40!yellow!60!white] (-8.25,-3) rectangle (-1.85,1.25);
        \node at (-6,-2.5) {$\mathsf{FPT}$ in $k$};
        
        \draw[thick,rounded corners] (-9,-3.4) rectangle (11,6.3);    
        \draw [thick,rounded corners] (-9,-3.7) rectangle (11,-9.5);
        \node at (10,-6.5) {\rotatebox{90}{No precedence constraints}};
        \node at (10,1.5) {\rotatebox{90}{Precedence constraints}};
        
        \draw [rounded corners, fill = white!100!black] (-8,2) rectangle (-2,3) node[pos=.5] {$1 | r_j,prec, p_j=1|k\text{-sched},C_{\max}$ \, \result{[A]}};
        
        \draw [rounded corners, fill = white!100!black] (-8,4 ) rectangle (-2,5) node[pos=.5] {$1 | prec, p_j=1|k\text{-sched},C_{\max}$ \, \result{[A]}};
        
        \draw [rounded corners, fill = white!100!black] (2,4) rectangle (8,5) node[pos=.5] {$1 |d_j, prec, p_j=1|k\text{-sched},C_{\max}$ \, \result{[B]}};
        
        \draw [rounded corners, fill = white!100!black] (2,2) rectangle (8,3) node[pos=.5] {$1 | prec|k\text{-sched},C_{\max}$ \, \result{[D]}};
        
        \draw [rounded corners, fill = white!100!black] (-8,0) rectangle (-2,1) node[pos=.5] {$P | prec, p_j=1|k\text{-sched},C_{\max}$\, \result{[C]}};
        
        \draw [rounded corners, fill = white!100!black] (-8,-2) rectangle (-2,-1) node[pos=.5] {$P |r_j, prec, p_j=1|k\text{-sched},C_{\max}$ \, \result{[C]}};
        
        \draw [rounded corners, fill = white!100!black] (2,-2) rectangle (8,-1) node[pos=.5] {$P | prec|k\text{-sched},C_{\max}$ \, \result{[D]}};
        
        \draw [rounded corners, fill = white!100!black] (2,0) rectangle (8,1) node[pos=.5] {$P |d_j, prec, p_j=1|k\text{-sched},C_{\max}$ \, \result{[B]}};

        \draw[->,thick] (-5,3.75) to[] (-5,3.25);
        \node at (-5.35,3.5) {$r_j$};
        
        \draw[->,thick] (-1.75,4.5) to[] (1.75,4.5);
        \node at (0,4.75) {$d_j$};
        
        \draw[->,thick] (-1.75,4.5) to[] (1.75,2.5);
       \node at (0.4,3.75) {arb $p_j$};
        
        
        \draw[->,thick] (-5,-0.25) to[] (-5,-0.75);
        \node at (-5.35,-0.5) {$r_j$};
        
        \draw[->,thick] (-1.75,0.5) to[] (1.75,-1.5);
        \node at (0.4,-0.25) {arb $p_j$};
        
        \draw[->,thick] (-1.75,0.5) to[] (1.75,0.5);
        \node at (0,0.75) {$d_j$};
        
        \draw[dashed] (-9,1.5) -- (9.5,1.5);
        \node at (9,3.5) {\rotatebox{90}{One machine}};
        \node at (9,-0.5) {\rotatebox{90}{Parallel machines}};
        

    \end{tikzpicture}
    
    \caption{An illustration of the various result types as indicated in Table~\ref{tab:tricho}. Arrows indicate how a problem is generalized by another problem.}
    \label{overview}
\end{figure}

We now describe some of our most significant technical contributions for obtaining the various types (listed as \result{[A]}-\result{[G]} in Table~\ref{tab:tricho}) of results. Note that we skip some less interesting cases in this introduction; for a complete argumentation of all results from Table~\ref{tab:tricho} we refer to Appendix~\ref{sec:cases}.  
The main building blocks and logical implications to obtain the results from Table~\ref{tab:tricho} are depicted in Figure~\ref{overview}. We now discuss these building blocks of Figure~\ref{overview} in detail.

\subparagraph{Precedence Constraints.}

Our main technical contribution concerns result type \result{[C]}. The simplest of the two cases, $P|\text{prec}, p_j=1|k\text{-sched},C_{\max}$, cannot be solved in $\mathcal{O}^*(2^{o(\sqrt{k \log k})})$ time assuming the Exponential Time Hypothesis and not in $2^{o(k)}$ unless sub-exponential time algorithms for the \textsc{Biclique} problem exist, due to reductions by Jansen et al.~\cite{jansen2016precedence}.
Our contribution lies in the following theorem that gives an upper bound for the more general of the two problems that matches the latter lower bound:
\begin{theorem}\label{thm:main}
 $P|r_j,\text{prec},p_j=1|k\text{-sched}, C_{\max}$ can be solved in $\mathcal{O}(8^kk(|V|+|E|))$ time,\footnote{We assume basic arithmetic operations with the release dates take constant time.} where $G = (V,E)$ is the precedence graph given as input.
\end{theorem}

Theorem~\ref{thm:main} will be proved in Section~\ref{Chapter:DP}. The first idea behind the proof is based on a natural\footnote{A similar dynamic programming approach was also present in for example~\cite{CyganPPW14}.} dynamic programming algorithm indexed by anti-chains of the partial order naturally associated with the precedence constraints. However, evaluating this dynamic program na\"ively would lead to an $n^{\mathcal{O}(k)}$ time algorithm, where $n$ is the number of jobs.

Our key idea is to only compute a subset of the table entries of this dynamic programming algorithm, guided by a new parameter of an antichain called the \emph{depth}.
Intuitively, the depth of an antichain $A$ indicates the number of jobs that can be scheduled after $A$ in a feasible schedule without violating the precedence constraints.

We prove Theorem~\ref{thm:main} by showing we may restrict attention in the dynamic programming algorithm to antichains of depth at most~$k$, and by bounding the number of antichains of depth at most~$k$ indirectly by bounding the number of \emph{maximal} antichains of depth at most~$k$. We believe this methodology should have more applications for scheduling problems with precedence constraints.

Surprisingly, the positive result of Theorem~\ref{thm:main} is in \emph{stark contrast} with the seemingly symmetric case where only deadlines are present: Our next result, indicated as \result{[B]} in Figure~\ref{overview} shows it is much harder:

\begin{theorem}\label{thm:wh}
 $P|d_j,prec,p_j=1|k\text{-sched}, C_{\max}$ is $\mathsf{W}[1]$-hard, and cannot be solved in $n^{o(k / \log k)}$ time assuming the ETH.
\end{theorem}

Theorem~\ref{thm:wh} is a consequence of a reduction outlined in Section~\ref{sec:lbprec}.
Note the $\mathsf{W}[1]$-hardness follows from a natural reduction from the $k$-{\sc Clique} problem (presented originally by Fellows and McCartin~\cite{fellows2003parametric}), but this reduction increases the parameter $k$ to $\Omega(k^2)$ and would only exclude~$n^{o(\sqrt{k})}$ time algorithms assuming the ETH. To obtain the tighter bound from Theorem~\ref{thm:wh}, we instead provide a non-trivial reduction from the $3${\sc-Coloring} problem based on a new selection gadget.

For result type~\result{[D]}, we give a lower bound by a (relatively simple) reduction from \textsc{Partitioned Subgraph Isomorphism} in Theorem~\ref{Thm:1|prec,rj|k-Cmax} and Corollary~\ref{Cor:2|prec|k-Cmax}. Since it is conjectured that \textsc{Partitioned Subgraph Isomorphism
} cannot be solved in $n^{o(k)}$ time assuming the ETH, our reduction is a strong indication that the simple $n^{\mathcal{O}(k)}$ time algorithm (see Appendix~\ref{sec:cases}) cannot be improved significantly in this case.

\subparagraph{No Precedence Constraints.}
The second half of our classification concerns scheduling problems without precedence constraints, and is easier to obtain than the first half.
Results \result{[E], [F]} are consequences of a greedy algorithm and Moore's algorithm \cite{moore1968n} that solves the problem $1||\sum_jU_j$ in $\mathcal{O}(n\log n)$ time. Notice that this also solves the problem~$1|r_j|k\text{-sched},C_{\max}$, by reversing the schedule and viewing the release dates as the deadlines.
For result type \result{[G]} we show that a standard technique in parameterized complexity, the color coding method, can be used to get a $2^{\mathcal{O}(k)}$ time algorithm for the most general problem of the class, being $R|r_j,d_j|k\text{-sched},C_{\max}$.
All lower bounds on the run time of algorithms for problems of type \result{[G]} are by a reduction from \textsc{Subset Sum}, but for $1|r_j,d_j|k\text{-sched},C_{\max}$ this reduction is slightly different.

\subsection{Related Work}

The interest in parameterized complexity of scheduling problems recently witnessed an explosive growth, resulting in e.g. a workshop~\cite{lorentz} and a survey by Mnich and van Bevern~\cite{mnich2018parameterized} with a wide variety of open problems.

The parameterized complexity of partial scheduling parameterized by the number of processed jobs, or equivalently, the number of jobs `on time' was studied before: Fellows et al.~\cite{fellows2003parametric} studied a problem called $k$-{\sc Tasks On Time} that is equivalent to $1| d_j,prec, p_j=1|k\text{-sched},C_{\max}$ and showed that it is $\mathsf{W}[1]$-hard when parameterized by $k$,\footnote{Our results $\result{[C]}$ and $\result{[D]}$ build on and improve this result.} and $\mathsf{FPT}$ parameterized by $k$ and the width of the partially ordered set induced by the precedence constraints.
Van Bevern et al.~\cite{van2015interval} showed that the {\sc Job Interval Selection} problem, where each job is given a set of possible intervals to be processed on, is $\mathsf{FPT}$ in $k$.
Bessy et al.~\cite{DBLP:journals/scheduling/BessyG19} consider partial scheduling with a restriction on the jobs called `Coupled-Task', and also remarked the current parameterization is relatively understudied. 

Another related parameter is the number of jobs that are \emph{not scheduled}, that also has been studied in several previous works~\cite{DBLP:journals/orl/BodlaenderF95,fellows2003parametric,DBLP:journals/mp/MnichW15}. For example, Mnich and Wiese~\cite{DBLP:journals/mp/MnichW15} studied the parameterized complexity of scheduling problems with respect to the number of rejected jobs in combination with other variables as parameter. If $n$ denotes the number of given jobs, this parameter equals $n-k$.
The two parameters are somewhat incomparable in terms of applications: In some settings only few jobs out of many alternatives need to be scheduled, but in other settings rejecting a job is very costly and thus will happen rarely. However, a strong advantage of using $k$ as parameter is in terms of its computational complexity: If the version of the problem with all jobs mandatory is $\mathsf{NP}$-complete it is trivially $\mathsf{NP}$-complete for~$n-k=0$, but it may still be $\mathsf{FPT}$ in $k$.

\subsection{Organization of this paper}
This paper is organized as follows: We start with some preliminaries in Section~\ref{sec:prel}.
In Section~\ref{Chapter:DP} we present the proof of Theorem~\ref{thm:main}, and in Section~\ref{sec:lbprec} we describe the reductions for result types \result{[B]} and \result{[D]}.
In Section~\ref{sec:colcoding} we give the algorithm for result type \result{[G]} and in Section~\ref{sec:conc} we present a conclusion. In Appendix~\ref{sec:omittedproofs3} we give the proofs omitted in Section~\ref{Chapter:DP}. The Lemma's and Theorems with these omitted proofs are indicated with a $\dagger$.  Finally, in Appendix~\ref{sec:cases} we motivate all cases from Table~\ref{tab:tricho}.

\section{Preliminaries: The three-field notation by Graham et al.}\label{sec:prel} 
Throughout this paper we denote scheduling problems using the three-field classification  by Graham et al.~\cite{graham1979optimization}. Problems are classified by parameters $\alpha | \beta | \gamma$. The $\alpha$ describes the machine environment. This paper uses $\alpha \in \{1,P,R\}$, indicating whether there are one ($1$), identical ($P$) or unrelated ($R$) parallel machines available. Here identical refers to the fact that every job takes a fixed amount of time process independent of the machine, and unrelated means a job could take different time to process per machine. The $\beta$ field describes the job characteristics, which in this paper can be a combination of the following values: $prec$ (precedence constraints), $r_j$ (release dates), $d_j$ (deadlines) and $p_j =1$ (all processing times are $1$). We assume without loss of generality that all release dates and deadlines are integers. 

The $\gamma$ field concerns the optimization criteria. A given schedule determines $C_j$, the completion time of job $j$, and $U_j$, the unit penalty which is $1$ if $C_j > d_j$, and $0$ if $C_j \le d_j$. In this paper we use the following optimization criteria

\begin{itemize}
    \item $C_{\max}$: minimize the makespan (i.e. the maximum completion time $C_j$ of any job),
    \item $\sum_jU_j$: minimize the number of jobs that finish after their deadline,
    \item $k\text{-sched}$: maximize the number of processed jobs; in particular, process at least $k$ jobs. 
\end{itemize}

A schedule is said to be \emph{feasible} if no constraints (deadlines, release dates, precedence constraints) are violated.

\section{Result Type C: Precedence Constraints, Release Dates and Unit Processing Times} \label{Chapter:DP}
In this section we provide a fast algorithm for partial scheduling with release dates and unit processing times parameterized by the number $k$ of scheduled jobs (Theorem~\ref{thm:main}). There exists a simple, but slow, algorithm with runtime $\mathcal{O}^*(2^{k^2})$ that already proves that this problem is $\mathsf{FPT}$ in $k$: This algorithm branches $k$ times on jobs that can be processed next. If more than $k$ jobs are available at a step, then processing these jobs greedily is optimal. Otherwise, we can recursively try to schedule all non-empty subsets of jobs to schedule next, and a $\mathcal{O}^*(2^{k^2})$ time algorithm is obtained via a standard (bounded search-tree) analysis. To improve on this algorithm, we present a dynamic programming algorithm based on table entries indexed by antichains in the precedence graph $G$ describing the precedence relations. Such an antichain describes the maximal jobs already scheduled in a partial schedule. 
Our key idea is that, to find an optimal solution, it is sufficient to restrict our attention to a subset of all antichains.
This subset will be defined in terms of the \emph{depth} of an antichain. With this algorithm we improve the runtime to $\mathcal{O}(8^kk(|V|+|E|))$.

By binary search, we can restrict attention to a variant of the problem that asks whether there is a feasible schedule with makespan at most $C_{\max}$, for a fixed universal deadline $C_{\max}$.


\subparagraph*{Notation for Posets.}
Any precedence graph $G$ is a directed acyclic graph and therefore induces a partial order $\prec$ on $V(G)$. Indeed, if there is a path from $x$ to $y$, we let $x \preceq y$.
An \emph{antichain} is a set $A \subseteq V(G)$ of mutually incomparable elements. We say $A$ is \emph{maximal} if there is no antichain $A'$ with $A \subset A'$.
The set of \emph{predecessors} of $A$ is $\pred(A) = \{x \in V(G): \exists a\in A: x \preceq a \}$, and the the set of \emph{comparables} of $A$ is $\comp(A) = \{x \in V(G): \exists a\in A:  x \preceq a \text{ or } x \succeq a\}$.
Note $\comp(A)=V(G)$ if and only if $A$ is maximal.

An element $x\in V(G)$ is a \emph{minimal} element if $x \preceq y$ for all $y \in \comp(\{x\})$. An element $x\in V(G)$ is a \emph{maximal} element if $x\succeq y$ for all $y \in \comp(\{x\})$. Furthermore $\min(G) =\break\{x\mid x\text{ is a minimal element in } G\}$ and $\max(G) = \{x\mid x\text{ is a maximal element in } G\}$.

Notice that $\max(G)$ is exactly the antichain $A$ such that $\pred(A) = V(G)$.  We denote the subgraph of $G$ induced by $S$ with $G[S]$. We may assume that $r_j < r_{j'}$ if $j\prec j'$ since job $j'$ will be processed later than $r_j$ in any schedule.
To handle release dates we use the following:

\begin{definition}
Let $G$ be a precedence graph. Then $G^t$ is the precedence graph restricted to all jobs that can be scheduled on or before time $t$, i.e. all jobs with release date at most $t$. 
\end{definition}
We assume $G = G^{C_{\max}}$, since all jobs with release date greater than~$C_{\max}$ can be ignored. 
\subparagraph*{The Algorithm.}
We now introduce our dynamic programming algorithm for $P|r_j,prec,p_j=1|k\text{-sched}, C_{\max}$. Let $m$ be the number of machines available.
We start with defining the table entries. For a given antichain $A \subseteq V(G)$ and integer $t$ we define
$$S(A,t) =  \begin{cases}1, & \text{ if there exists a feasible schedule of makespan } t \text{ that processes } \pred(A)  ,\\
 0, & \text{ otherwise.}
\end{cases}$$
Computing the values of $S(A,t)$ can be done by trying all combinations of scheduling at most $m$ jobs of $A$ at time $t$ and then checking whether all remaining jobs of $\pred(A)$ can be scheduled in makespan $t-1$.
To do so, we also verify that all the jobs in $A$ actually have a release date at or before $t$.
Formally, we have the following recurrence for $S(A,t)$:

\begin{lemma}\label{lemmaS}
$$S(A,t) = (A \subseteq V(G^{t})) \wedge \bigvee_{X\subseteq A: |X| \le m}S(A',t-1) : A' = \max(\pred(A)\setminus X).$$
\end{lemma} 

\begin{proof}
If $A\not\subseteq V(G^t)$, then there is a job $j \in A$ with $r_j > t$. And thus $S(A,t)=0$.

For any $X \subseteq A$, $X$ is a set of maximal elements with respect to $G[\pred(A)]$, and consists of pair-wise incomparable jobs, since $A$ is an antichain. So, we can schedule all jobs from $X$ at time $t$ without violating any precedence constraints. Define $ A' = \max(\pred(A)\setminus X)$ as the unique antichain such that $\pred(A)\setminus X = \pred(A')$. If $S(A',t-1)=1$ and $|X|\le m$, we can extend the schedule of $S(A',t-1)$ by scheduling all $X$ at time $t$. In this way we get a feasible schedule processing all jobs of $\pred(A)$ before or at time $t$. So if we find such an $X$ with $|X|\le m$ and $S(A',t-1)=1$, we must have $S(A,t)=1$. 

For the other direction, if for all $X\subseteq A$ with $|X|\le m$, $S(A',t-1)=0$, then no matter which set $X \subseteq A$ we try to schedule at time $t$, the remaining jobs cannot be scheduled before~$t$. Note that only jobs from $A$ can be scheduled at time $t$, since those are the maximal jobs. Hence, there is no feasible schedule and $S(A,t) = 0$.
\end{proof}

The above recurrence cannot be directly evaluated, since the number of different antichains of a graph can be big: there can be as many as $\binom{n}{k}$ different antichains with $|\pred(A)|\le k$, for example in the extreme case of an independent set. Even when we restrict our precedence graph to have out degree $k$, there could be $k^k$ different antichains, for example in $k$-ary trees. To circumvent this issue, we restrict our dynamic programming algorithm only to a specific subset of antichains. To do this, we use the following new notion of the \emph{depth} of an antichain.

\begin{definition}
Let $A$ be an antichain. Define the depth (with respect to $t$) of $A$ as
$$d^t(A) = |\pred(A)| + |\min(G^t- \comp(A))|.$$
We also denote $d(A) = d^{C_{\max}}(A)$.
\end{definition}

\begin{figure}
    \centering
    \begin{tikzpicture} [scale = 0.65, transform shape]

\tikzset{vertex/.style = {shape=circle,draw,minimum size=1.5em}}
\tikzset{edge/.style = {->,> = latex'}}
\node[vertex, fill = red!90!black] (a) at  (0,0) {};
\node[vertex, fill = black] (b) at  (4,2) {};
\node[vertex, fill = cyan!60!white, very thick] (c) at  (4,0) {};
\node[vertex, fill = cyan!60!white, very thick] (d) at  (4,-2) {};

\node[vertex] (e) at  (8,4) {};
\node[vertex] (f) at  (8,3) {};
\node[vertex] (g) at  (8,2) {};

\node[vertex, fill = cyan!60!white] (h) at  (8,1) {};
\node[vertex, fill = cyan!60!white] (i) at  (8,0) {};
\node[vertex, fill = cyan!60!white] (j) at  (8,-1) {};

\node[vertex, fill = cyan!60!white] (k) at  (8,-2) {};
\node[vertex, fill = cyan!60!white] (l) at  (8,-3) {};
\node[vertex, fill = cyan!60!white] (m) at  (8,-4) {};

\draw[edge] (a) to (b);
\draw[edge] (a) to (c);
\draw[edge] (a) to (d);
\draw[edge] (b) to (e);
\draw[edge] (b) to (f);
\draw[edge] (b) to (g);
\draw[edge] (c) to (h);
\draw[edge] (c) to (i);
\draw[edge] (c) to (j);
\draw[edge] (d) to (k);
\draw[edge] (d) to (l);
\draw[edge] (d) to (m);

\node[vertex, fill = black] (e) at  (12,1) {};
\node[anchor=west] at (12.5,1) {\Large = job in antichain $A$};
\node[vertex, fill = red!90!black] (f) at  (12,0) {};
\node[anchor=west] at (12.5,0) {\Large = job in $\pred(A)$};
\node[vertex, fill = cyan!60!white, very thick] (g) at  (12,-1) {};
\node[anchor=west] at (12.5,-1) {\Large = job in $\min(G- \comp(A))$};
\node[vertex, fill = cyan!60!white] (g) at  (12,-2) {};
\node[anchor=west] at (12.5,-2) {\Large = job in $G- \comp(A)$};
\node[vertex] (h) at  (12,-3) {};
\node[anchor=west] at (12.5,-3) {\Large = job in $\comp(A)$};

\node[] at (8,-5) {\Large $d(A) = |\pred(A)| + |\min(G- \comp(A))| = 2 + 2$};

\end{tikzpicture}
    \caption{Example of an antichain and its depth in a perfect $3$-ary tree. We see that $|\pred(A)|=2$, but $d(A)=4$. If $k=2$, the dynamic programming algorithm will not compute $S(A,t)$ since $d(A)>k$. The only antichains with depth $\le 2$ are the empty set and the root node $r$ on its own as a set. Indeed $d(\emptyset) = d(\{r\}) =1$. Note that for instances with $k=2$, a feasible schedule may exist. If so, we will find that $R(\{r\},1)=1$, which will be defined later. In this way, we can still find the antichain $A$ as a solution.}

    \label{k-tree}
\end{figure}
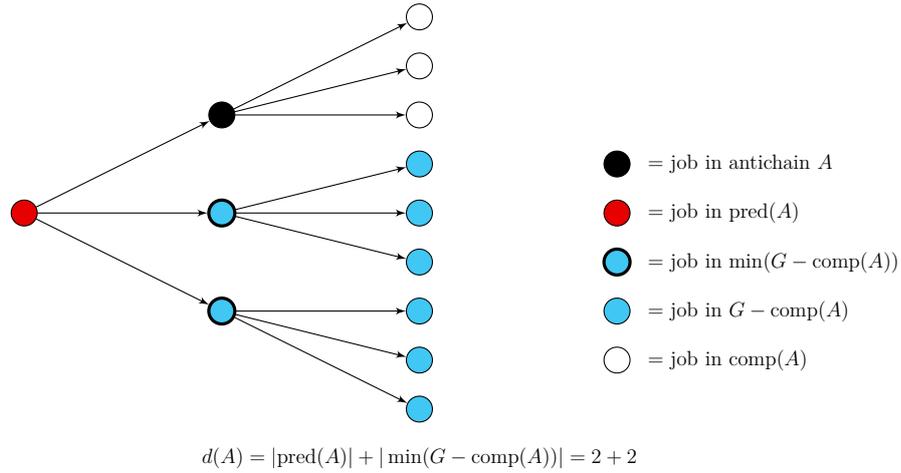

The intuition behind this definition is that it quantifies the number of jobs that can be scheduled before (and including) $A$ without violating precedence constraints. See Figure~\ref{k-tree} for an example of an antichain and its depth. We restrict the dynamic programming algorithm to only compute $S(A,t)$ for $A$ satisfying $d^t(A) \le k$. This ensures that we do not go `too deep' into the precedence graph unnecessarily at the cost of a slow runtime.

Because of this restriction in the depth, it could happen that we check no antichains with $k$ or more predecessors, while there are corresponding feasible schedules. It is therefore possible that for some antichains $A$ with $d^t(A)>k$, there is a feasible schedule for all $\ge k$ jobs in $\pred(A)$ before time $C_{\max}$, but the value $S(A,C_{\max})$ will not be computed. To make sure we still find an optimal schedule, we also compute the following condition $R(A,t)$ for all~$t\le C_{\max}$ and antichains $A$ with $d^t(A)\le k$:
$$R(A,t) = \begin{cases} 
1, & \text{if there exists a feasible schedule  with makespan at most } C_{\max} \text{ that}\\
&  \text{processes }\pred(A) \text{ on or before $t$ and processes jobs from }\\
& \min(G - \pred(A)) \text{ after $t$, with a total of } k \text{ jobs processed},\\
0, & \text{otherwise.}
\end{cases}$$

By definition of $R(A,t)$, if $R(A,t)=1$ for any $A$ and $t\le C_{\max}$, then we find a feasible schedule that processes $k$ jobs on time.\footnote{The reverse direction is more difficult and postponed to Lemma~\ref{lem:corr}.}
We show in Appendix~\ref{sec:omittedproofs3} that $R(A,t)$ can be quickly computed: 

\begin{lemma}[$\dagger$]\label{lemmaR}
There is an $\mathcal{O}(|V|k + |E|)$ time algorithm $\mathtt{fill(A,t)}$ that, given an antichain $A$, integer $t$, and value $S(A,t)$, computes $R(A,t)$.
\end{lemma}

The algorithm $\mathtt{fill(A,t)}$ checks if $S(A,t)=1$ and if so, greedily schedules jobs from $\min(G-\pred(A))$ after $t$ in order of smallest release date. If $k - |\pred(A)|$ jobs can be scheduled before $C_{\max}$, it returns `true' ($R(A,t)=1$). Otherwise, it returns `false' ($R(A,t)=0)$. 

Combining all steps gives us the algorithm as described in Algorithm~\ref{algorithmDP}. It remains to bound its runtime and argue its correctness.

 \begin{algorithm}
\SetAlgoLined
\ForEach{$t=1,...,C_{\max}$}{
    Enumerate all antichains $A$ in $G^t$ with $d^t(A) \le k$ using Lemma~\ref{cor:nrantichainsfort}\\
    \ForEach{antichain $A$ in $G^t$ with $d^t(A) \le k$}{
        Compute $S(A,t)$ using Lemma~\ref{lemmaS}\\
        \If{$\mathtt{fill}(S(A,t),A,t)$}{
            \Return{TRUE}
        }
    }
}
\Return{FALSE}
\caption{Algorithm for $P|\pred,p_j=1|k\text{-sched},C_{\max}$}
\label{algorithmDP}
\end{algorithm}

\subparagraph*{Runtime.}
To analyze the runtime of the dynamic programming algorithm, we need to bound the number of checked antichains. Recall that we only check antichains $A$ with $d^t(A)\le k$ for each time $t\le C_{\max}$. We first analyze the number of antichains $A$ with $d(A)\le k$ in any graph and use this to upper bound the number of antichains  checked at time $t$.

\begin{lemma}[$\dagger$] \label{cor:nrantichainsfort}
	For any $t$, there are at most $4^k$ antichains $A$ with $d^t(A)\le k$ in any precedence graph $G=(V,E)$, and they can be enumerated within $\mathcal{O}(4^k(|V|+|E|))$ time.
\end{lemma}

Notice that to compute each $S(A,t)$, we look at a maximum of $\binom{k}{m} \le 2^k$ different sets~$X$. Computing the antichain $A'$ such that $A' = \max(\pred(A)\setminus X)$ takes $\mathcal{O}(|V|+|E|)$ \break time. After this computation, $R(A,t)$ is directly computed in $\mathcal{O}(|V|k + |E|)$ time. For each time~$t~\in~\{1,...,C_{\max}\}$, there are at most $4^k$ different antichains $A$ for which we compute~$S(A,t)$ and $R(A,t)$. Since $C_{\max}\leq k$, we therefore have total runtime of $\mathcal{O}(4^kk(2^k(|V|+|E|)+(|V|k+|E|)))$.
Hence, Algorithm~\ref{algorithmDP} runs in time $\mathcal{O}(8^kk(|V|+|E|))$.

\subparagraph*{Correctness of algorithm.}
To show that the algorithm described in Algorithm~\ref{algorithmDP} indeed returns the correct answer, the following lemma is clearly sufficient:
\begin{lemma}[$\dagger$]\label{lem:corr} A feasible schedule for $k$ jobs with makespan at most $C_{\max}$ exists if and only if $R(A,t)=1$ for some $t\le C_{\max}$ and antichain $A$ with $d^t(A)\le k$.
\end{lemma}

To prove Lemma~\ref{lem:corr}, we consider the schedule which corresponds to an antichain which has minimal depth. We then conclude that it either should be witnessed by some $R(A,t)$ or that there is another antichain with even smaller depth, which contradicts the assumption. The proof heavily relies on the intricacies of the definition of depth.


\section{Result Types B and D: One Machine and Precedence Constraints}\label{sec:lbprec}
In this section we show that Algorithm \ref{algorithmDP} cannot be even slightly generalized further: if we allow job-dependent deadlines or non-unit processing times, the problem becomes $\mathsf{W}[1]$-hard parameterized by $k$ and cannot be solved in $n^{o(k / \log k)}$ time unless the ETH fails.

\subparagraph*{Job-dependent deadlines.}
The fact that combining precedence constraints with job-\break dependent deadlines makes the problem $\mathsf{W}[1]$-hard, is a direct consequence from the fact that $1| \textit{prec}, p_j=1|\sum_jU_j$ is $\mathsf{W}[1]$-hard, parameterized by $n-\sum_jU_j = k$ where $n$ is the number of jobs \cite{fellows2003parametric}. It is important to notice that the notation of these problems implies that each job can have its own deadline. Hence, we conclude from this that $1| d_j,\textit{prec}, p_j=1 |k\text{-sched},C_{\max}$ is $\mathsf{W}[1]$-hard parameterized by $k$. This is a reduction from $k$-{\sc Clique} and therefore we get a lower bound on algorithms for the problem of $n^{\Omega(\sqrt{k})}$. Based on the Exponential Time Hypothesis, we now sharpen this lower bound with a reduction from $3${\sc-Coloring}:

\begin{theorem}\label{lem:ethbound} $1| d_j,\textit{prec}, p_j=1 |k\text{-sched},C_{\max}$ is $\mathsf{W}[1]$-hard parameterized by $k$. Furthermore, there is no algorithm solving $1| d_j,\textit{prec}, p_j=1 |k\text{-sched},C_{\max}$ in $2^{o(n)}$ time where $n$ is the number of jobs, assuming ETH.
\end{theorem}
\begin{proof}
	The proof will be a reduction from $3${\sc-Coloring}, for which no $2^{o(|V|+|E|)}$ algorithm exists under the Exponential Time Hypothesis \cite[pages 471-473]{cygan2015parameterized}. Let the graph $G=(V,E)$ be the instance of $3${\sc-Coloring} with $|V|=n'$ and $|E|=m'$. We then create the following instance for $1|d_j,prec,p_j=1|k\text{-sched},C_{\max}$. 
	\begin{itemize}
		\item For each vertex $v_i \in V$, create $6$ jobs:
		\begin{itemize}
			\item $v_i^1$, $v_i^2$  and $v_i^3$ with deadline $d_{v_i} = i$,
			\item $w_i^1$, $w_i^2$  and $w_i^3$ with deadline $d_{w_i} = n'+2m'+1-i$,
		\end{itemize}
		add precedence constraints $v_i^1 \prec w_i^1$, $v_i^2 \prec w_i^2$ and $v_i^3 \prec w_i^3$. These jobs represent which color for each vertex will be chosen (if $v_i^1$ and $w_i^1$ are processed, vertex $i$ gets color $1$).
		\item For each edge $e_j \in E$, create $12$ jobs:
		\begin{itemize}
			\item $e_j^{12}$, $e_j^{13}$, $e_j^{21}$, $e_j^{23}$, $e_j^{31}$ and $e_j^{32}$ with deadline $d_{e_j} =n'+j$,
			\item $f_j^{12}$, $f_j^{13}$, $f_j^{21}$, $f_j^{23}$, $f_j^{31}$ and $f_j^{32}$ with deadline $d_{f_j} =n'+m'+1-j$,
		\end{itemize}  
		add precedence constraints $e_j^{ab} \prec f_j^{ab}$. These jobs represent what the colors of the endpoints of an edge will be. So if the jobs $e_j^{ab}$ and $f_j^{ab}$ are processed for $e=\{u,v\}$, then vertex $u$ has color $a$ and vertex $v$ has color $b$. Since the endpoints should have different colors, the jobs $e_j^{aa}$ and $f_j^{aa}$ do not exist. 
		\item For each $e_j^{ab}$ with $e=\{u,v\}$ add the precedence constraints $u^a \prec e_j^{ab}$ and $v^b \prec e_j^{ab}$.
		\item Set $C_{\max} = k = 2n'+2m'$. 
	\end{itemize}
	
	We now prove that the created instance is a yes instance if and only if the original $3${\sc-Coloring} instance is a yes instance. 
	Assume that there is a $3$-coloring of the graph $G=(V,E)$. Then there is also a feasible schedule: For each vertex $v_i$ with color $a$, process the jobs $v_i^a$ and $w_i^a$ at their respective deadlines. For each edge $e_j=\{u,v\}$ with $u$ colored $a$ and $v$ colored $b$, process the jobs $e_j^{ab}$ and $f_j^{ab}$ exactly at their respective deadlines. Notice that because it is a $3$-coloring, each edge has endpoints of different colors, so these jobs exist. Also note that no two jobs were processed at the same time. Exactly $2n'+2m'$ jobs were processed before time $2n'+2m'$. Furthermore, no precedence constraints were violated. 
	
	For the other direction, assume that we have a feasible schedule in our created instance of $1| d_j,\textit{prec}, p_j=1 |k\text{-sched},C_{\max}$. Let $\mathcal{V}_i = \{v_i^1,v_i^2,v_i^3\}$,  $\mathcal{W}_i = \{w_i^1,w_i^2,w_i^3\}$,  and let $\mathcal{E}_j = \{e_j^{12},e_j^{13}, e_j^{21},e_j^{23},e_j^{31},e_j^{32}\}$ and $\mathcal{F}_j = \{f_j^{12},f_j^{13}, f_j^{21},f_j^{23},f_j^{31},f_j^{32}\}$. We show by induction on $i$ that out of each of the sets $\mathcal{V}_i$, $\mathcal{W}_i$, $\mathcal{E}_j$ and $\mathcal{F}_j$, exactly one job was scheduled at its deadline. 
	
	Since we have a feasible schedule, at time $2m'+2n'$ one of the jobs of~$\mathcal{W}_1$ must be scheduled, since they are the only jobs with a deadline greater than~$2n+2m-1$. However, if $w_1^a$ was scheduled at time $2m'+2n'$, then the job $v_1^a$ must be processed at time~$1$ because of precedence constraints and since its deadline is $1$. 
	Note, that no other jobs from $\mathcal{V}_1$ and $\mathcal{W}_1$ can be processed, due to their deadlines and precedence constraints.  
		
	Now assume that all sets $\mathcal{V}_1,...,\mathcal{V}_{i-1},\mathcal{W}_1,...,\mathcal{W}_{i-1}$ have exactly one job scheduled at their respective deadline, and no more can be processed. Since we have a feasible schedule, one job should be scheduled at time $2n'+2m'-(i-1)$. However, since no more jobs from~$\mathcal{W}_1,...,\mathcal{W}_{i-1}$ can be scheduled, the only possible jobs are from $\mathcal{W}_i$ since they are the only other jobs with a deadline greater than $2n'+2m'-i$. However, if $w_i^a$ was scheduled at time $2n'+2m'-(i-1)$, then the job $v_i^a$ must be processed at time $i$ because of precedence constraints, its deadline at $i$ and because at times $1,...,i-1$ other jobs had to be processed. Also, no other job from $\mathcal{V}_i$ can be processed in the schedule, since they all have deadline $i$. As a consequence, no other jobs from $\mathcal{W}_1$ can be processed, as they are restricted to precedence constraints. So the statement holds for all set $\mathcal{V}_i$ and $\mathcal{W}_i$. In the exact same way, one can conclude the same about all sets $\mathcal{E}_j$ and $\mathcal{F}_j$. 
	
	Because of this, we see that each job and each vertex have received a color from the schedule. They must form a $3$-coloring, because a job from $\mathcal{E}_j$ could only be processed if the two endpoints got two different colors. Hence the $3${\sc-Coloring} instance is a yes instance. 
	
	As $k = 2n'+2m'$ we therefore conclude there is no $2^{o(n)}$ algorithm under the ETH. 
\end{proof}

Note that this bound significantly improves the old lower bound of $2^{\Omega(\sqrt{n})}$ implied by the the reduction from $k$-{\sc Clique} reduction: Since $k \leq n$, Theorem~\ref{lem:ethbound} implies that

\begin{corollary}\label{cor:lb3col}
Assuming ETH, there is no algorithm solving $1| d_j,\textit{prec}, p_j=1 |k\text{-sched},C_{\max}$ in $n^{o(k/\log(k))}$ where $n$ is the number of jobs.
\end{corollary}

\subparagraph*{Non-unit processing times.}
We show that having non-unit processing times combined with precedence constraints make the problem $\mathsf{W}[1]$-hard even on one machine. The proof of Theorem~\ref{thm:nonunit} heavily builds on the reduction from $k$-{\sc Clique} to $k$-{\sc Tasks On Time} by Fellows and McCartin~\cite{fellows2003parametric}.
\begin{theorem}\label{thm:nonunit}
$1|prec|k\text{-sched},C_{\max}$ is $\mathsf{W}[1]$-hard, parameterized by $k$. 
\end{theorem}

\begin{proof}
	The proof is a reduction from $k$-{\sc Clique}. We start with $G=(V,E)$, an instance of $k$-{\sc Clique}. For each vertex $v \in V$, create a job $j_v$ with $p_{j_v} = 2$. For each edge $e \in E$, create a job $j_e$ with $p_{j_e}=1$. Now for each edge $(u,v)$, add the following two precedence relations: $j_u \prec j_e$ and $j_v \prec j_e$, so before one can process a job associated with an edge, both jobs associated with the endpoints of that edge need to be finished. Now let~$k' = k + \frac{1}{2}k(k-1)$ and $C_{\max} = 2k + \frac{1}{2}k(k-1)$. We will now prove that~$1|prec|k'\text{-sched},C_{\max}$ is a yes instance if and only of $k$-{\sc Clique} is a yes instance.
	
	Assume that the $k$-{\sc Clique} instance is a yes instance, then process first the $k$ jobs associated with the vertices of the $k$-{clique}. Next process the $\frac{1}{2}k(k-1)$ jobs associated with the edges of the $k$-{clique}. In total, $k+\frac{1}{2}k(k-1)=k'$ jobs are now processed with a makespan of $2k + \frac{1}{2}k(k-1)$. Hence, the instance of $1|prec|k'\text{-sched},C_{\max}$ is a yes instance.
	
	For the other direction, assume $1|prec|k'\text{-sched},C_{\max}$ to be a yes instance, so we have found a feasible schedule. For any feasible schedule, if one schedules $l$ jobs associated with vertices, then at most $\frac{1}{2}l(l-1)$ jobs associated with edges can be processed, because of the precedence constraints. However, because $k'=k+\frac{1}{2}k(k-1)$ jobs were done in the feasible schedule before $C_{\max} = 2k+\frac{1}{2}k(k-1)$, at most $k$ jobs associated with vertices can be processed, because they have processing time of size $2$. Hence, we can conclude that exactly $k$ vertex-jobs and $\frac{1}{2}k(k-1)$ edge-jobs were processed. Hence, there were $k$ vertices connected through $\frac{1}{2}k(k-1)$ edges, which is a $k$-clique. 
\end{proof}

The proofs of Theorem~\ref{Thm:1|prec,rj|k-Cmax} and Corollary~\ref{Cor:2|prec|k-Cmax} are reductions from \textsc{Partitioned Subgraph Isomorphism}. Let $P=(V',E')$ be a `pattern' graph, $G= (V,E)$ be a `target' graph, and $\chi : V \to V'$ a `coloring' of the vertices of $G$ with elements from $P$. A \emph{$\chi$-colorful $P$-subgraph of $G$} is a mapping $\varphi: V' \rightarrow V$ such that (1) for each $\{u,v\} \in E'$ it holds that $\{\varphi(u),\varphi(v)\} \in E$ and (2) for each $u \in V'$ it holds that $\chi(\varphi(u))=u$. If $\chi$ and $G$ are clear from the context they may be omitted in this definition.

\begin{definition}[\textsc{Partitioned Subgraph Isomorphism}]
	Given graphs $G =(V,E)$ and $P=(V',E')$, $\chi : V \to V'$. Determine whether there is a $\chi$-colorful $P$-subgraph of $G$.
\end{definition}

\begin{theorem}[Marx \cite{marx2010can}]
	\textsc{Partitioned Subgraph Isomorphism} cannot be solved in $n^{o(|E'|/ \log |E'|)}$ time assuming the Exponential Time Hypothesis (ETH).
\end{theorem}

We will now reduce \textsc{Partitioned Subgraph Isomorphism} to $1|\text{prec},r_j|k\text{-sched},C_{\max}$.
\begin{theorem} \label{Thm:1|prec,rj|k-Cmax}
	$1|\text{prec},r_j|k\text{-sched},C_{\max}$ cannot be solved in $n^{o(k / \log k)}$ time assuming the Exponential Time Hypothesis (ETH). 
\end{theorem}
\begin{proof}
	Let $G = (V,E)$, $P = (V',E')$ and $\chi : V \to V'$. We will write $V' = \{1,\dots,s\}$. Define for $i=0,\dots,s$ the following important time stamps: $$t_i := \sum_{j=1}^i 3^{s+1-j}.$$ 
	Construct the following jobs for the instance of the $1|\text{prec},r_j|k\text{-sched},C_{\max}$ problem:
	\begin{itemize}
		\item For $i=1,\dots,s$:
		\begin{itemize}
			\item For each vertex $v \in V$ such that $\chi(v) =i$, create a job $j_v$ with processing time $p(j_v)= 3^{s+1 - i}$ and release date $t_{i-1}$.
		\end{itemize}
		\item For each $(v,w) \in E$ such that $(\chi(v),\chi(w))\in E'$, create a job $j_{v,w}$ with $p(j_{v,w}) = 1$ and release date $t_s$. Add precedence constraints $j_v \prec j_{v,w}$ and $j_w \prec j_{v,w}$.
	\end{itemize}
	Then ask whether there exists a solution to the scheduling problem for $k = s + |E'|$ with makespan $C_{\max} \le t_s + |E'|$.  
	
	Let the \textsc{Partitioned Subgraph Isomorphism} instance be a yes-instance and let $\varphi: V(P) \to V(G)$ be a colorful $P$-subgraph. We claim the following schedule is feasible:
	
	\begin{itemize}
		\item For $i=1,\dots,s$:
		\begin{itemize}
			\item Process $j_{\varphi(i)}$ at its release date $t_{i-1}$.
		\end{itemize}
		\item Process for each $(i,i') \in E'$ the job $j_{\varphi(i),\varphi(i')}$ somewhere in the interval $[t_s,t_s+|E'|]$.
	\end{itemize}
	Notice that all jobs are indeed processed after their release date and that in total there are $k =s + |E'|$ processed before $C_{\max} \le t_s+ |E'|$. 
	Furthermore, all precedence constraints are respected as any edge job is processed after both its predecessors. 
	Also, the edge jobs $e^{\varphi(i),\varphi(i')}$ must exist, as $\varphi(P)$ is a properly colored $P$-subgraph. 
	Therefore, we can conclude that indeed this schedule is feasible.
	
	For the other direction, assume that there is a solution to the created instance of $1|\text{prec},r_j|k\text{-sched},C_{\max}$. 
	Define $J_i = \{j_v : \chi(v) = i\}$. 
	We will first prove that at most $1$ job from each set $J_i$ can be processed in a feasible schedule. 
	To do this, we first prove that at most $1$ job from each set $J_i$ can be processed before $t_s$.
	Any job in $J_i$ has release date $t_{i-1} = \sum_{j=1}^{i-1} 3^{s+1-j}$. Therefore, there is only $t_s-t_{i-1} = \sum_{j=i}^s 3^{s+1-j}$ time left to process the jobs from $J_i$ before time $t_s$. However, the processing time of any job in $J_i$ is $3^{s+1-i}$, and since $2\cdot 3^{s+1-i} > \sum_{j=i}^s 3^{s+1-j}$, at most $1$ job from $J_i$ can be processed before $t_s$. 
	Since all jobs not in some $J_i$ have their release date at $t_s$, at most $s$ jobs are processed at time $t_s$. Thus at time $t_s$, there are $|E'|$ time unit left to process $|E'|$ jobs, because of the choice of $k$ and makespan. 
	Hence the only way to get a feasible schedule is to process exactly one job from each set $J_i$ at its respective release date and process exactly $|E'|$ edge jobs after $t_s$. 
	
	Let $v^i$ be the vertex, such that $j_v$ was processed in the feasible schedule with color $i$. 
	We will show that $\varphi:V(P) \to V(G)$, defined as $\varphi(i) = v^i$, is a function such that $\varphi(P)$ is a properly colored $P$-subgraph of $G$. 
	Hence, we are left to prove that for each $(i,i') \in E'$, the edge $(\varphi(i),\varphi(i')) \in E$, i.e. that for each $(i,i') \in E'$, the job $j_{\varphi(i),\varphi(i')}$ was processed. 
	Because only the vertex jobs $j_{\varphi(1)}, j_{\varphi(2)}, \dots, j_{\varphi(s)}$ were processed, the precedence constraints only allow for edge jobs $j_{\varphi(i),\varphi(i')}$ to be processed. 
	We created edge job $j_{v,w}$ if and only if $(v,w) \in E$ and $(\chi(v),\chi(w)) \in E'$, hence the $|E'|$ edge jobs have to be exactly the edge jobs $j_{\varphi(i),\varphi(i')}$ for $(i,i')\in E'$. 
	Therefore, we proved indeed that $\varphi(P)$ is a colorful $P$-subgraph of $G$.
	
	Notice that $k= s+|E'| \le 3|E'|$ as we may assume the number of vertices in $P$ is at most $2|E'|$. Hence the given bound follows.
\end{proof}

\begin{corollary}\label{Cor:2|prec|k-Cmax}
	$2|\text{prec}|k\text{-sched},C_{\max}$ cannot be solved in $n^{o(k / \log k)}$ time assuming the Exponential Time Hypothesis (ETH). 
\end{corollary}
\begin{proof}
	We can use the same idea for the reduction from \textsc{Partitioned Subgraph Isomorphism} as in the proof of Theorem~\ref{Thm:1|prec,rj|k-Cmax}, except for the release dates, as they are not allowed in this type of scheduling problem. 
	To simulate the release dates, we use the second machine as a release date machine, meaning that we will create a job for each upcoming release date and will require these new jobs to be processed. 
	More formally: For $i=1,\dots,s$, create a job $j_{r_i}$ with processing time $3^{s+1-i}$ and precedence constraints $j_{r_i} \prec j$ for any job $j$ that had release date $t_i$ in the original reduction. Furthermore let $j_{r_i} \prec j_{r_{i+1}}$. 
	Then we add $|E'|$ jobs $j'$ with processing time $1$ and with precedence relations $j_{r_s} \prec j'$. We then ask whether there exists a feasible schedule with $k = 2s + 2|E'|$ and with makespan $t_s + |E'|$. All newly added jobs are required in any feasible schedule and therefore, all other arguments from the previous reduction also hold. Finally, note that $k$ is again linear in $|E'|$. 
\end{proof}

\section{Result Type G: $k$-scheduling without Precedence Constraints}
\label{sec:colcoding}
The problem $P|k\text{-sched}|C_{\max}$, cannot be solved in $2^{o(k)}$ time assuming the ETH by a reduction to \textsc{Subset Sum}.
 We show that the problem is fixed-parameter tractable with a matching run time in $k$, even in the case of unrelated machines, release dates and deadlines, denoted by $R|r_j,d_j,k\text{-sched}|C_{\max}$.

\begin{theorem}
	$R|r_j,d_j,k\text{-sched}|C_{\max}$ is fixed-parameter tractable in $k$ and can be solved in $\mathcal{O}^*((2e)^kk^{\mathcal{O}(\log k)})$ time.
\end{theorem}

\begin{proof}
	We give an algorithm that solves any instance of $R|r_j,d_j,k\text{-sched}|C_{\max}$ within $\mathcal{O}^*((2e)^kk^{\mathcal{O}(\log k)})$ time. The algorithm is a randomized algorithm that can be de-randomized using the color coding method, as described by Alon et al. \cite{alon1995color}. The algorithm first (randomly) picks a coloring $c : \{1,...,n\} \to \{1,...,k\}$, so each job is given one of the $k$ available colors. We then compute whether there is a feasible colorful schedule, i.e. a feasible schedule that processes exactly one job of each color. If this colorful schedule can be found, then it is possible to schedule at least $k$ jobs before $C_{\max}$. 
	
	Given a coloring $c$, we compute whether there exists a colorful schedule in the following way. Define for $ 1 \le i \le m$ and $X\subseteq \{1,...,k\}$:  
	\begin{align*}B_i(X) = &\text{minimum makespan of all schedules on machine } i \text{ processing } |X| \text{ jobs,}\\ &\text{each from a different color in } X.
	\end{align*}
	Clearly $B_i(\emptyset) = 0$, and all values $B_i(X)$ can be computed in $\mathcal{O}(2^k n)$ time  using the following:
	\begin{lemma}\label{lemmab}
	Let $\min\{\emptyset\} = \infty$. Then
		$$B_i(X) = \min_{l \in X}\min_{j: c(j)=l}\{ C_j =\max\{r_j, B_i(X\setminus\{l\})\} + p_{ij} : C_j \le d_j\}.$$ 
	\end{lemma}
	\begin{proof}
		In a schedule on one machine with $|X|$ jobs using all colors from $X$, one job should be scheduled as last, defining the makespan. So for all possible jobs $j$, we compute what the minimal end time would be if $j$ was scheduled at the end of the schedule. This $j$ cannot start before its release date or before all other colors are scheduled.
	\end{proof}
	Next, define for $ 1 \le i \le m$ and~$X\subseteq [k]$, $A_i(X)$ to be $1$ if  $B_i(X) \le  C_{\max}$, and to be $0$ otherwise.	
	So $A_i(X)=1$ if and only if $|X|$ jobs, each from a different color of $X$, can be scheduled on machine $i$ before $C_{\max}$. A colorful feasible schedule exists if and only if there is some partition $X_1,...,X_m$ of $\{1,..,k\}$ such that $\Pi_{i=1}^m A_i(X_i) = 1$. The \emph{subset convolution} of two functions is defined as $(A_i * A_{i'}) (X) = \sum_{Y\subseteq X} A_i(Y) A_{i'}(X\setminus Y)$. Then $\Pi_{i=1}^m A_i(X_i) = 1$ if and only if $(A_1*\cdots*A_m)(\{1,...,k\}) >0$. The value of $(A_1*\cdots*A_m)(\{1,...,k\}) >0$ can be computed in $2^k k^{\mathcal{O}(1)}$ time using fast subset convolution \cite{bjorklund2007fourier}. 
	
	An overview of the randomized algorithm is given in Algorithm \ref{algorithm1}. If the $k$ jobs that are processed in an optimal solution are all in different colors, the algorithm outputs true. By standard analysis, $k$ jobs are all assigned different colors with probability at least $1/e^k$, and thus $e^k$ independent trials to boost the error probability of the algorithm to at most $1/2$.
	
	\begin{algorithm}[h!] 
		\SetAlgoLined
		For a given coloring $c$:\\
		\ForEach{$i = 1,...,m$}{
			\ForEach{$X \subseteq \{1,..,k\}$ in order of increasing size}{
				Compute $B_i(X)$ using Lemma \ref{lemmab}. \\
				Set $A_i(X)=1$ if $B_i(X)\le C_{\max}$, set $A_i(X) = 0$ otherwise.
		}}
		Compute $(A_1*\cdots*A_m)(\{1,...,k\})$ using fast subset convolution \cite{bjorklund2007fourier}. \\
		\If{$(A_1*\cdots*A_m)(\{1,...,k\})>0$}{\Return{TRUE}} 
		\caption{Algorithm for solving $R|r_j,d_j,k\text{-sched}|C_{\max}$}
		\label{algorithm1}
	\end{algorithm}
	
	By using the standard methods by Alon et al.~\cite{alon1995color}, Algorithm~\ref{algorithm1} can be derandomized.
\end{proof}


\section{Concluding Remarks}\label{sec:conc}
We classify all studied variants of partial scheduling parameterized by the number of jobs to be scheduled to be either in $\mathsf{P}$, $\mathsf{NP}$-complete and fixed-parameter tractable by $k$, or $\mathsf{W}[1]$-hard parameterized by $k$.
Our main technical contribution is an $\mathcal{O}(8^kk(|V|+|E|))$ time algorithm for $P|r_j,\text{prec},p_j=1|k\text{-sched}, C_{\max}$.

In a fine-grained sense, the cases we left open are cases 3-20 from Table~\ref{tab:tricho}.
We believe in fact algorithms in rows 5-6 and 10-20 are optimal: An $n^{o(k)}$ time algorithm for any case from result type $\result{[C]}$ or $\result{[D]}$ would imply either a $2^{o(n)}$ time algorithm for \textsc{Biclique} or an $n^{o(k)}$ time algorithm for \textsc{Partitioned Subgraph Isomorphism}, which both would be surprising.
It would be interesting to see whether for any of the remaining cases with precedence constraints and unit processing times a `sub-exponential' time algorithm exists.



A related case is $P3|\text{prec},p_j=1|C_{max}$ (where $P3$ denotes three machines). It is a famously hard open question (see e.g.~\cite{GareyJ79}) whether this can be solved in polynomial time, but maybe it is doable to try to solve this question in sub-exponential time, e.g. $2^{o(n)}$?



\bibliography{biblio}
\appendix


\section{Omitted Proofs from Section~\ref{Chapter:DP}}\label{sec:omittedproofs3}

\subsection{Correctness $\mathtt{fill(A,t)}$ sub-procedure: Proof of Lemma~\ref{lemmaR}}
Remember that the algorithm $\mathtt{fill(A,t)}$ checks if $S(A,t) = 1$ and if so, greedily schedules jobs from $\min(G - \pred(A))$ after $t$ in order of smallest release date. If~$k - |\pred(A)|$ jobs can be scheduled before $C_ {\max}$, it returns ‘true’ ($R(A,t) = 1$). Otherwise, it returns ‘false’ ($R(A,t) = 0$). 

First we show that if $\mathtt{fill(A,t)}$ returns `true', it follows that $R(A,t)=1$. Since $S(A,t)=1$, all jobs from $\pred(A)$ can be finished at time $t$. Take that feasible schedule and process $k-|\pred(A)|$ jobs from $\min(G-\pred(A))$ between $t$ and $C_{\max}$. This is possible because $\mathtt{fill(A,t)}$ is true. All predecessors of jobs in $\min(G-\pred(A))$ are in $\pred(A)$ and therefore processed before $t$. Hence, no precedence constraints are violated and we find a feasible schedule with the requirements, i.e. $R(A,t)=1$. 

For the other direction, assume that $R(A,t)=1$, i.e. we find a feasible schedule $\sigma$ where exactly the jobs from $\pred(A)$ are processed on or before $t$ and only jobs from $\min(G-\pred(A))$ are processed after $t$. Thus $S(A,t)=1$.
Define $M$ as the set of jobs processed after $t$ in $\sigma$. If~$M$ equals the set of jobs with the smallest release dates of $\min(G-\pred(A))$, we can also process the jobs of $M$ in order of increasing release dates. Then $\mathtt{fill(A,t)}$ will be `true', since $M$ has size at least $k-|\pred(A)|$. However, if $M$ is not that set, we can replace a job which does not have one of the smallest $k-|\pred(A)|$ release dates, by one which has and was not in $M$ yet. This new set can then still be processed between $t+1$ and $C_{\max}$ because smaller release dates impose weaker constraints. We keep replacing until we end up with $M$ being exactly the set of jobs with smallest release dates, which is then proved to be schedulable between $t$ and $C_{\max}$. Hence, $\mathtt{fill(A,t)}$ will return `true'.

Computing the set $\min(G - \pred(A))$ can be done in $\mathcal{O}(|V| + |E|)$ time. Sorting them on release date can be done in $\mathcal{O}(|V|k)$ time, as there are at most $k$ different release dates. Finally, greedily scheduling the jobs while checking feasibility can be done in $\mathcal{O}(|V|)$ time. Hence this algorithm runs in time $\mathcal{O}(|V|k + |E|)$.

\subsection{Bound on number of antichains: Proof of Lemma~\ref{cor:nrantichainsfort} }

To analyze the number of antichains $A$ with $d(A)\le k$, we give an upper bound on this number via an upper bound on the number of maximal antichains. Recall from the notations for posets, that for a maximal antichain $A$ we have $\comp(A) = V(G)$, and therefore $d(A) = |\pred(A)|$. The following lemma connects the number of antichains and maximal antichains of bounded depth:

\begin{lemma}\label{lem:dmax}
	For any antichain $A$, there exists a maximal antichain $A_{\max}$ such that $A\subseteq A_{\max}$ and $d(A)=d(A_{\max})$.
\end{lemma}
\begin{proof}
	Let $A_{\max} = A \cup \min(G- \comp(A)) $. By definition, all elements in $\min(G- \comp(A))$ are incomparable to each other and incomparable to any element of $A$. Hence $A_{\max}$ is an antichain. Since $\comp(A_{\max}) = V(G)$, $A_{\max}$ is a maximal antichain. Moreover, $d(A)=~|\pred(A)| + |\min(G- \comp(A))| = |\pred(A_{\max})| = d(A_{\max})$, since the elements in $\min(G- \comp(A))$ are minimal elements and all their predecessors are in $\pred(A)$ besides themselves.
\end{proof}

By Lemma~\ref{lem:dmax}, we see that all antichains are a subset of maximal antichains with the same depth. For any (maximal) antichain $A$ with $d(A)\le k$, $|A|\le k$. Hence we observe:

\begin{corollary}
	$$|\{A:A\text{ antichain}, d(A)\le k\}| \le 2^k|\{A:A\text{ maximal antichain}, d(A)\le k\}|.$$ 
\end{corollary}

This corollary allows us to restrict attention to only upper bounding the number of \emph{maximal} antichains of bounded depth.

\begin{lemma} \label{lamma:maxantichains}
	There are at most $2^k$ maximal antichains $A$ with $d(A)\le k$ in any precedence graph $G = (V,E)$, and they can be enumerated in $\mathcal{O}(2^kk(|V|+|E|))$ time.
\end{lemma}

\begin{proof}
	Let $\mathcal{A}_k(G)$ be the set of maximal antichains in $G$ with depth at most~$k$. We prove that $|\mathcal{A}_k(G)|\le 2^k$ for any graph $G$ by induction on $k$. Clearly $ |\mathcal{A}_0(G)|\le 1$ for any graph $G$, since the only antichain with $d(A) \le 0$ is $A = \emptyset$ if $G=\emptyset$.
	
	Let $k >0$ and assume $|\mathcal{A}_j(G)| \le 2^j$ for $j < k$ for any graph $G$. If we have a precedence graph $G$ with minimal elements $s_1,...,s_l$, we partition $\mathcal{A}_k(G)$ into $l+1$ different sets~$\mathcal{B}_1, \mathcal{B}_2,...,\mathcal{B}_{l+1}$. The set $\mathcal{B}_i$ is defined as the set of maximal antichains $A$ of depth at most~$k$ in which $s_1,...,s_{i-1}\subseteq A$, but $s_i \not \in A$. If $s_i \not \in A$, then $s_i \in \pred(A)$ since $A$ is maximal, so any such maximal antichain has a successor of $s_i$ in $A$. If we  define $S_j$ as the set of all successors of $s_j$ (including $s_j$), we see that $\mathcal{B}_i = \mathcal{A}_{k-i} \left(G - \left(\bigcup_{j=1}^{i-1}S_j \cup \{s_i\}\right)\right)$. Indeed, if~$A \in \mathcal{B}_i$, then $\{s_1,...,s_{i-1}\} \subseteq A$. Hence we can remove those elements and its successors from the graph, as they are comparable to any such antichain. Moreover, we can also remove~$s_i$ (but \emph{not} its successors) from the graph, since it is in $\pred(A)$. Thus $\mathcal{B}_i$ is then exactly the set of maximal antichains with depth $i$ less in the remaining graph. We get the following recurrence relation: 
	\begin{equation} 
	|\mathcal{A}_k(G)| = \sum_{i=1}^l \left|\mathcal{A}_{k-i} \left(G - \left(\bigcup_{j=1}^{i-1}S_j \cup \{s_i\}\right)\right)\right| +1,
	\label{recurrence}
	\end{equation} 
	since $|\mathcal{B}_{l+1}|$, the number of antichains satisfying $\{s_1,...,s_l \}\subseteq A$, is exactly one. Notice that we may assume that $l\le k$, because otherwise the depth of the antichain will be greater than~$k$. Then if we use the induction hypothesis that $|\mathcal{A}_j(G)| \le 2^j$ for $j < k$ for any graph~$G$, we see by (\ref{recurrence}) that:
	\begin{align*}
		|\mathcal{A}_k(G)| &= \sum_{i=1}^l \left|\mathcal{A}_{k-i} \left(G - \left(\bigcup_{j=1}^{i-1}S_j \cup \{s_i\}\right)\right)\right| +1, \\
		&\le 2^k \left(\sum_{i=1}^k \frac{1}{2^i} + \frac{1}{2^k} \right)\\
		&= 2^k.
	\end{align*}
	
	The lemma follows since the above procedure can easily be modified in an recursive algorithm to enumerate the antichains, and by using a Breadth-First Search we can compute $ G - \left(\bigcup_{j=1}^{i-1}S_j \cup \{s_i\}\right)$ in $O(|V|+|E|)$ time. Thus, each recursion step takes $\mathcal{O}(k(|V|+|E|))$ time. 
\end{proof}

Returning to (non-maximal) antichains, we find the following upper bound:

\begin{corollary} \label{cor:nrantichains}
	There are at most $4^k$ antichains $A$ with $d(A)\le k$ in any precedence graph $G=(V,E)$, and they can be enumerated within $\mathcal{O}(4^k(|V|+|E|))$ time.
\end{corollary}
Notice that the runtime is indeed correct, as it dominates both the time needed for the construction of the set $\mathcal{A}_k(G)$ and the time needed for taking the subsets of $\mathcal{A}_k(G)$ (which is~$2^k|\mathcal{A}_k(G)|$).

We now restrict the number of antichains $A$ in $G^t$ with $d^t(A) \le k$. Take $G^t$ to be the graph in \autoref{cor:nrantichains} and notice that $d^t(A) = d(A)$ for any antichains $A$ in $G^t$. By Corollary~\ref{cor:nrantichains} we obtain Lemma~\ref{cor:nrantichainsfort}.


\subsection{Correctness algorithm: Proof of Lemma~\ref{lem:corr}}
To prove Lemma~\ref{lem:corr}, we need one more definition.

\begin{definition}
Let $\sigma$ be a feasible schedule. Then $A(\sigma)$ is the antichain such that $\pred(A(\sigma))$ is exactly the set of jobs that was scheduled in $\sigma$. 
\end{definition}
 Equivalently, if $X$ is the set of jobs processed by $\sigma$, then $A(\sigma)=\max(G[X])$.

Clearly, if $R(A,t)=1$ for some $t\le C_{\max}$ and antichain $A$ with $d^t(A)\le k$, we have a feasible schedule with $k$ jobs by definition of $R(A,t)$. Hence, it remains to prove that if a feasible schedule for $k$ jobs exists, then $R(B,t)=1$ for some $t\le C_{\max}$ and antichain $B$ with~$d^t(B)\le k$. Let 
$$\Sigma^* = \{\sigma | \sigma \text{ is a feasible schedule that processes } k \text{ jobs and has a makespan of at most } C_{\max}\},$$
so $\Sigma^*$ is the set of all possible solutions.
Define
$$\sigma^* = \underset{\sigma}{\text{argmin}}\{d(A(\sigma)) | \sigma \in \Sigma^*\}, $$
i.e. $\sigma^*$ is a schedule for which $A(\sigma^*)$ has minimal depth (with respect to $C_{\max}$). We now define $t$ and $B$ such that $R(B,t) =1$.

\begin{itemize}
    \item Let $t=\max\{t: $ job not in $ \max(G[\pred(A(\sigma^*))])$ was scheduled at time  $t \}$, so from $t+1$ and on, only maximal jobs (with respect to $G[\pred(A(\sigma^*))]$)  are scheduled.
    \item Let $M = \{x:$ job $x$ was scheduled at $t+1$ or later  in $\sigma^*$ $\}$.
    \item Let $B = \max(\pred(A(\sigma^*))\setminus M)$, so $\pred(B)$ is exactly the set of jobs scheduled on or before time $t$ in $\sigma^*$.
\end{itemize}

See Figure~\ref{fig:sigma*} for an illustration of these concepts. There are two cases to distinguish:

\subparagraph{$\mathbf{d^t(B)\le k}$.} In this case we prove that $R(B,t) =1$. The feasible schedule we are looking for in the definition of $R(B,t)$ is exactly $\sigma^*$. Indeed, all jobs from $\pred(B)$ were finished at time~$t$. Furthermore, all jobs in $M$ are maximal, so all their predecessors are in $\pred(B)$. Hence, $M\subseteq \min(G - \pred(B))$. So, by definition $R(B,t)=1$. 

\subparagraph{$\mathbf{d^t(B)> k}$.} In this case we prove that there is a schedule $\sigma'$ such that $d(A(\sigma')) < d(A(\sigma^*))$, i.e. we find a contradiction to that fact that $d(A(\sigma^*))$ was minimal. 
    This $\sigma'$ can be found as follows: take schedule $\sigma^*$ only up until time $t$. Let $C$ be a subset of $\min(G^t-\comp(B))$ such that $|C| = k - |\text{pred(B)}|$. This $C$ can be found since $d^t(B)\ge k$. Process the jobs in $C$ after time $t$ in $\sigma'$. These can all be processed without precedence constraint or release date violations, since their predecessors were already scheduled and $C\subseteq G^t$. So, we find a feasible schedule that processes $k$ jobs, called $\sigma'$. The choice of $\sigma'$ is depicted in \autoref{fig:sigma'}. Note that~$C\subseteq \min(G^t-\comp(B)) \subseteq \min(G-\comp(B))$ and not all jobs of $\min(G-\comp(B))$ are necessarily processed in $\sigma'$.

    It remains to prove that $d(A(\sigma')) < d(A(\sigma^*))$. Define $D(A) = \pred(A) \cup \min(G-\comp(A))$ for any antichain $A$. So $D(A)$ is the set of jobs that contribute to $d(A)$ and so $|D(A)|=d(A)$. We will prove that $D(B) = D(A(\sigma')) \subset D(A(\sigma^*))$. This will be done in two steps, first we show that $D(B) = D(A(\sigma')) \subseteq D(A(\sigma^*))$. In the last step we prove $D(B) \neq D(A(\sigma^*))$, which gives us $d(A(\sigma')) < d(A(\sigma^*))$.
    
    Notice that $C\subseteq D(B)$ since $C \subseteq \min(G-\comp(B))$, hence $D(B) = D(B\cup C)$. Since~$A(\sigma') = B \cup C$ it follows that $D(A(\sigma')) = D(B)$. Next we prove that $D(B) \subseteq D(A(\sigma^*)).$
    Clearly, if $x \in \pred(B)$ then $x \in \pred(A(\sigma^*))$.
    It remains to show that $x\in  \min(G-\comp(B))$ implies that $x\in  D(A(\sigma^*))$.
    If $x\in  \min(G-\comp(B))$, then either $x \in M$  or $x \not \in M$. If $x \in M$, then $x\in A(\sigma^*)$ so $x \in  \pred(A(\sigma^*))$. If $x \not \in M$, then $x \not \in \comp(B\cup M)$ since $x$ was a minimal element in $\min(G-\comp(B))$. Since $A(\sigma^*)\subseteq B\cup M$, and thus $\comp(A(\sigma^*)) \subseteq \comp(B \cup M)$,  we observe that $x \in \min(G-\comp(A(\sigma^*)))$. We then conclude that $D(B)\subseteq D(A(\sigma^*))$.
    
    We are left to show that $D(B) \neq D(A(\sigma^*))$. Remember that $t$ was chosen such that there is a job processed at time $t$ that was not in $\max(G[\pred(A(\sigma^*))])$. In other words, there was a job $x \in B$ in $\sigma^*$ at time $t$ with $y\in M$ such that $y \succ x$. Note that $y\not \in D(B)$, since $y\in M$, so $y$ is not in $\pred(B)$ and $y$ is clearly comparable to $x$. However, $y\in D(A(\sigma^*))$, so we find that $d(A(\sigma')) =d(B) < d(A(\sigma^*))$. Hence, we found a schedule with smaller $d(A(\sigma'))$, which leads to a contradiction.

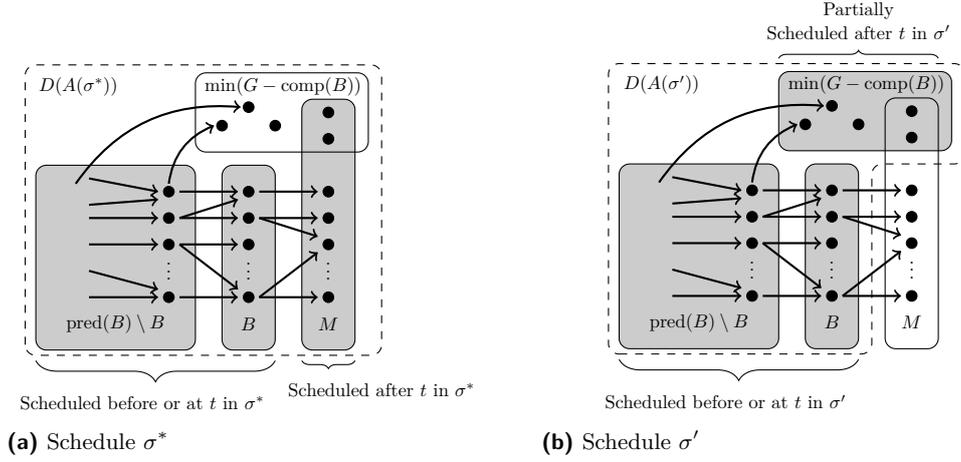
\begin{figure}
\centering
    \begin{subfigure}[b]{0.45\textwidth}
    \centering
     \begin{tikzpicture} [scale = 0.7, transform shape]
        \draw [rounded corners, fill = white!80!black] (4.5,-1) rectangle (5.5,2.5);
        \draw [fill] (5,0) circle [radius=0.1];
        \node at (5,0.6) {$\vdots$};
        \draw [fill] (5,1) circle [radius=0.1];
        \draw [fill] (5,1.5) circle [radius=0.1];
        \draw [fill] (5,2) circle [radius=0.1];
        \node at (5,-0.5) {$B$};
        
         \draw [rounded corners, fill = white!80!black] (1,-1) rectangle (4,2.5);
         \node at (2.5,-0.5) {$\pred(B)\setminus B$};       
        \draw [fill] (3.5,0) circle [radius=0.1];
        \node at (3.5,0.6) {$\vdots$};
        \draw [fill] (3.5,1) circle [radius=0.1];
        \draw [fill] (3.5,1.5) circle [radius=0.1];
        \draw [fill] (3.5,2) circle [radius=0.1];
        
        \draw [rounded corners, fill = white!80!black] (6,-1) rectangle (7,3.75);
        \node at (6.5,-0.5) {$M$};
        
        \draw [rounded corners] (4,2.75) rectangle (7.25,4.25);
        \node at (5.65,4) {$\min(G-\comp(B))$};

        \draw [fill] (6.5,0) circle [radius=0.1];
        \node at (6.5,0.6) {$\vdots$};
        \draw [fill] (6.5,1) circle [radius=0.1];
        \draw [fill] (6.5,1.5) circle [radius=0.1];
        \draw [fill] (6.5,2) circle [radius=0.1];
        \draw [fill] (6.5,3) circle [radius=0.1];
        \draw [fill] (6.5,3.5) circle [radius=0.1];

        \draw[->,thick] (5.2,2) to (6.3,2);
        \draw[->,thick] (5.2,0) to (6.3,0.85);
        \draw[->,thick] (5.2,1.5) to (6.3,1.5);
        \draw[->,thick] (5.2,0) to (6.3,0);       
        \draw[->,thick] (5.2,1.5) to (6.3,1.15);
        
        \draw[->,thick] (3.7,1) to (4.8,1); 
        \draw[->,thick] (3.7,2) to (4.8,2);
        \draw[->,thick] (3.7,1.5) to (4.8,1.85);
        \draw[->,thick] (3.7,1.5) to (4.8,1.5);
        \draw[->,thick] (3.7,0) to (4.8,0);       
        \draw[->,thick] (3.7,1) to (4.8,0.15);

        \draw[->,thick] (2,1) to (3.3,1); 
        \draw[->,thick] (2,2.25) to (3.3,2);
        \draw[->,thick] (2,1.75) to (3.3,1.85);
        \draw[->,thick] (2,1.5) to (3.3,1.5);
        \draw[->,thick] (2,0) to (3.3,0);       
        \draw[->,thick] (2,0.5) to (3.3,0.15);        
        
        \draw [fill] (5.5,3.25) circle [radius=0.1];
        \draw [fill] (4.5,3.25) circle [radius=0.1];   
        \draw [fill] (5,3.6) circle [radius=0.1];       
        
        \draw[->,thick] (3.5,2.15) to[bend left] (4.35,3.25);
        \draw[->,thick] (1.75,2.15) to[bend left] (4.85,3.6);

        \draw[decorate,decoration={brace,mirror,raise=4pt,amplitude=8pt}]  (1,-1)--(5.5,-1) ;
        \node at (3,-2) {Scheduled before or at $t$ in $\sigma^*$};
    
        \draw[decorate,decoration={brace,mirror,raise=4pt,amplitude=3pt}] (6,-1)--(7,-1) ;
        \node at (7.5,-1.75) {Scheduled after $t$ in $\sigma^*$};
        
         \draw [rounded corners, dashed] (0.8,-1.10) rectangle (7.5,4.4);
         \node at (1.8,4) {$D(A(\sigma^*))$};

    \end{tikzpicture}
    
        \caption{Schedule $\sigma^*$}
        \label{fig:sigma*}
    \end{subfigure}
    \hspace{0.5 cm}
    \begin{subfigure}[b]{0.45\textwidth}
    \centering
     \begin{tikzpicture} [scale = 0.7, transform shape]
        \draw [rounded corners, fill = white!80!black] (4.5,-1) rectangle (5.5,2.5);
        \draw [fill] (5,0) circle [radius=0.1];
        \node at (5,0.6) {$\vdots$};
        \draw [fill] (5,1) circle [radius=0.1];
        \draw [fill] (5,1.5) circle [radius=0.1];
        \draw [fill] (5,2) circle [radius=0.1];
        \node at (5,-0.5) {$B$};
        
         \draw [rounded corners, fill = white!80!black] (1,-1) rectangle (4,2.5);
         \node at (2.5,-0.5) {$\pred(B)\setminus B$};       
        \draw [fill] (3.5,0) circle [radius=0.1];
        \node at (3.5,0.6) {$\vdots$};
        \draw [fill] (3.5,1) circle [radius=0.1];
        \draw [fill] (3.5,1.5) circle [radius=0.1];
        \draw [fill] (3.5,2) circle [radius=0.1];

        \draw [rounded corners, fill= white!80!black] (4,2.75) rectangle (7.25,4.25);
        \node at (5.65,4) {$\min(G-\comp(B))$};

        \draw [rounded corners] (6,-1) rectangle (7,3.75);
        \node at (6.5,-0.5) {$M$};
        
        \draw [fill] (6.5,0) circle [radius=0.1];
        \node at (6.5,0.6) {$\vdots$};
        \draw [fill] (6.5,1) circle [radius=0.1];
        \draw [fill] (6.5,1.5) circle [radius=0.1];
        \draw [fill] (6.5,2) circle [radius=0.1];
        \draw [fill] (6.5,3) circle [radius=0.1];
        \draw [fill] (6.5,3.5) circle [radius=0.1];

        \draw[->,thick] (5.2,2) to (6.3,2);
        \draw[->,thick] (5.2,0) to (6.3,0.85);
        \draw[->,thick] (5.2,1.5) to (6.3,1.5);
        \draw[->,thick] (5.2,0) to (6.3,0);       
        \draw[->,thick] (5.2,1.5) to (6.3,1.15);
        
        \draw[->,thick] (3.7,1) to (4.8,1); 
        \draw[->,thick] (3.7,2) to (4.8,2);
        \draw[->,thick] (3.7,1.5) to (4.8,1.85);
        \draw[->,thick] (3.7,1.5) to (4.8,1.5);
        \draw[->,thick] (3.7,0) to (4.8,0);       
        \draw[->,thick] (3.7,1) to (4.8,0.15);

        \draw[->,thick] (2,1) to (3.3,1); 
        \draw[->,thick] (2,2.25) to (3.3,2);
        \draw[->,thick] (2,1.75) to (3.3,1.85);
        \draw[->,thick] (2,1.5) to (3.3,1.5);
        \draw[->,thick] (2,0) to (3.3,0);       
        \draw[->,thick] (2,0.5) to (3.3,0.15);        
        
        \draw [fill] (5.5,3.25) circle [radius=0.1];
        \draw [fill] (4.5,3.25) circle [radius=0.1];   
        \draw [fill] (5,3.6) circle [radius=0.1];       
        
        \draw[->,thick] (3.5,2.15) to[bend left] (4.35,3.25);
        \draw[->,thick] (1.75,2.15) to[bend left] (4.85,3.6);

        \draw[decorate,decoration={brace,mirror,raise=4pt,amplitude=8pt}]  (1,-1)--(5.5,-1) ;
        \node at (3,-2) {Scheduled before or at $t$ in $\sigma'$};
    
        \draw[decorate,decoration={brace,raise=4pt,amplitude=3pt}] (4,4.3)--(7,4.3) ;
        \node at (5.5,5) {Scheduled after $t$ in $\sigma'$};
        \node at (5.5,5.4) {Partially};
        
        \draw [rounded corners, dashed] (0.8,-1.10) -- (5.75,-1.10) -- (5.75,2.5) -- (7.5,2.5) -- (7.5,4.4) -- (0.8,4.4) -- cycle;
        \node at (1.8,4) {$D(A(\sigma'))$};
        
    \end{tikzpicture}
    
        \caption{Schedule $\sigma'$}
        \label{fig:sigma'}
    \end{subfigure}
    
    \caption{Visualization of the definitions of $M$ and $B$ and the schedules $\sigma^*$ in the proof of Lemma~\ref{lem:corr} is shown in \autoref{fig:sigma*}. \autoref{fig:sigma'} depicts the schedule $\sigma'$ as chosen in the subcase $d(B)>k$. The grey boxes indicate which jobs are processed in the schedules. We will prove that $|D(A(\sigma'))|<|D(A\sigma^*))|$.}
    \label{correctness}
\end{figure}

\section{Explicit Motivation all cases}\label{sec:cases}
For completeness and the readers convenience, we explain in this section for each row of Table~\ref{tab:tricho} how the upper and lower bounds are obtained. 

First notice that the most general variant $R|r_j,d_j,prec|k\text{-sched},C_{\max}$ can be solved in $n^{\mathcal{O}(k)}$ time as follows:
Guess for each machine the set of jobs that are scheduled on it, and guess how they are ordered in an optimal solution, to get sequences $\sigma_1,\ldots,\sigma_m$ with a joint length equal to $k$.
For each such $(\sigma_1,\ldots,\sigma_m)$, run the following simple greedy algorithm to determine whether the minimum makespan achieved by a feasible schedule that schedules for each machine $i$ the jobs as described in $\sigma_i$: Iterate $t=1,\ldots,n$ and schedule the job $\sigma_i(t)$ at machine $i$ as early as possible without violating release dates/deadline and precedence constraints (if this is not possible, return NO). Since each optimal schedule can be assumed to be normalized in the sense that no single job can be executed earlier, it is easy to see that this algorithm always returns an optimal schedule for some choice of $\sigma_1,\ldots,\sigma_m$.
Since there are only $n^{\mathcal{O}(k)}$ different sequences $\sigma_1,\ldots,\sigma_m$ of combined length $k$, the run time follows.

\begin{description}
	\item[Cases 1-2:] The polynomial time algorithms behind result \result{[A]} are obtained by a straightforward greedy algorithm: For $1| r_j,prec, p_j=1|k\text{-sched},C_{\max}$, build the schedule from beginning to end, and schedule an arbitrary job if any is available; otherwise wait until one becomes available.
	\item[Cases 3-4, 7-8:] The given lower bound is by Corollary~\ref{cor:lb3col}. 
	\item[Cases 5-6:] The upper bound is by the algorithm of Theorem~\ref{thm:main}. The lower bound is due to reduction by Jansen et al. \cite{jansen2016precedence}. In particular, if no sub exponential time algorithm for the \textsc{Biclique} problem exist, there exist no algorithms in $n^{o(k)}$ time for these problems. 
	\item[Case 9:] The lower bound is by Theorem~\ref{thm:nonunit}, which is a reduction from $k$\textsc{-Clique} and heavily builds on the reduction from $k$\textsc{-Clique} to $k$-\textsc{Tasks On Time} by Fellows and McCartin~\cite{fellows2003parametric}. This reduction increases the parameter $k$ to $\Omega(k^2)$, hence the lower bound of $n^{o(\sqrt{k})}$.
	\item[Cases 10-20:] The given lower bound is by Theorem~\ref{Thm:1|prec,rj|k-Cmax}, which is a reduction from \textsc{Partitioned Subgraph Isomorphism}. It is conjectured that there exist no algorithms solving \textsc{Partitioned Subgraph Isomorphism} in $n^{o(k)}$ time assuming ETH, which would imply that the $n^{\mathcal{O}(k)}$ algorithm for these problems cannot be improved significantly.
	\item[Cases 21-28:] Result \result{[E]} is established by a simple greedy algorithm that always schedules an available job with the earliest deadline.
	\item[Cases 29-31:] Result \result{[F]} is a consequence of Moore's algorithm \cite{moore1968n} that solves the problem $1||\sum_jU_j$ in $\mathcal{O}(n\log n)$ time. The algorithm creates a sequence $j_1,\dots,j_n$ of all jobs in earliest due date order. It then repeats the following steps: It tries to process the sequence (in the given order) on one machine. Let $j_i$ be the first job in the sequence that is late. Then a job from $j_1,\dots,j_i$ with maximal processing time is removed from the sequence.
	If all jobs are on time, it returns  the sequence followed by the jobs that have been removed from the sequence.
	Notice that this also solves the problem~$1|r_j|k\text{-sched},C_{\max}$, by reversing the schedule and viewing the release dates as the deadlines.
	\item[Cases 32:] The lower bound for this problem is a direct consequence of the reduction from \textsc{Knapsack} to $1|r_j|\sum_j U_j$ by Lenstra et al.~\cite{lenstra1977complexity}, which is a linear reduction. Jansen et al.~\cite{jansen2016bounding} showed that \textsc{Subset Sum} (and thus also \textsc{Knapsack}) cannot be solved in $2^{o(n)}$ time assuming ETH.
	\item[Cases 33-40:] Since $2||C_{\max}$ is equivalent to \textsc{Subset Sum} and can therefore not be solved in $2^{o(n)}$ time assuming ETH, as showed by Jansen et al.\cite{jansen2016bounding}. Therefore, its generalizations, in particular those mentioned in cases 33-40, have the same lower bound on run times assuming ETH.
\end{description}

\end{document}